\documentclass{article}
\usepackage{romp}
\usepackage{latexsym}
\usepackage[english]{babel}
\usepackage{fancyhdr}
\usepackage[mathscr]{eucal}
\usepackage{amsmath}
\usepackage{mathrsfs}
\usepackage{amsthm}
\usepackage{amsfonts}
\usepackage{amssymb}
\usepackage{amscd}
\usepackage{bbm}
\usepackage{graphicx}
\usepackage{graphics}
\usepackage{latexsym}

\newcommand{\cH}{\mathcal{H}}

\newcommand{\ud}{\mathrm{d}}

\title{ Multiplicity of self-adjoint realisations \\ of the (2+1)-fermionic model \\ of Ter-Martirosyan--Skornyakov type }

\author{ Alessandro Michelangeli, \thanks{ \;\; Partially supported by the 2014-2017 MIUR-FIR grant ``\emph{Cond-Math: Condensed Matter and Mathematical Physics}'' code RBFR13WAET.}\\ SISSA -- International School for Advanced Studies \\
				Via Bonomea 265, 34136 Trieste (Italy) \\ e-mail: alemiche@sissa.it\\[2ex]\\
	Andrea Ottolini, \\ Department of Mathematics, Stanford University\\ 450 Serra Mall, Stanford CA 94305 \\ e-mail: ottolini@stanford.edu
	}
\begin{document}
\maketitle
\begin{center}\textbf{SISSA preprint 65/2016/MATE (2016)}
\end{center}
\bigskip\bigskip
\begin{abstract}
We reconstruct the whole family of self-adjoint Hamiltonians of Ter-Martirosyan--Skornyakov type for a system of two identical fermions coupled with a third particle of different nature through an interaction of zero range. We proceed through an operator-theoretic approach based on the self-adjoint extension theory of Kre{\u\i}n, Vi\v{s}ik, and Birman. We identify the explicit `Kre{\u\i}n-Vi\v{s}ik-Birman extension parameter' as an operator on the `space of charges' for this model (the `Kre{\u\i}n space') and we come to formulate a sharp conjecture on the dimensionality of its kernel. Based on our conjecture, for which we also discuss an amount of evidence, we explain the emergence of a multiplicity of extensions in a suitable regime of masses and we reproduce for the first time the previous partial constructions obtained by means of an alternative quadratic form approach.
\end{abstract}

 Singular perturbations of the Laplacian. Self-adjoint extensions. Kre{\u\i}n-Vi\v{s}ik-Birman theory. Ter-Martirosyan--Skornyakov operators. Fermionic models of zero-range interactions.

%\tableofcontents

\section{Introduction}

The so-called \emph{singular perturbations} of the multi-variable Laplacian %with action 
$\Delta\equiv \sum_{n=1}^N (\frac{1}{2m_j}\Delta_{x_n})$ on the Hilbert space $\cH_N:=\bigotimes_{n=1}^N L^2(\mathbb{R}^d,\ud x_n)$, with mass parameters $m_1,\dots,m_N>0$, are customarily referred to as self-adjoint extensions of the restrictions $\Delta\upharpoonright C^\infty_0(\mathbb{R}^{Nd}\setminus\Gamma_\mathscr{C})$: these are restrictions of $\Delta$ to smooth functions that are compactly supported away from the `singular manifold' $\Gamma_\mathscr{C}:=\bigcup_{(i,j)\in\mathscr{C}}\Gamma_{ij}$ associated with a given collection $\mathscr{C}$ of couples $(i,j)$ of distinct variables $i,j\in\{1,\dots,N\}$, where $\Gamma_{ij}:=\{x_i=x_j\}$ denotes the coincidence hyperplane between variable $x_i$ and $x_j$. In the applications the extension is taken in whole space $\cH_N$ or also, depending on the context, in a suitable Hilbert subspace of $\cH_N$ determined by the additional prescription of permutation symmetry or anti-symmetry on a sub-collections of variables.

In Quantum Mechanics singular perturbations of the (negative) Laplacian have the natural interpretation of quantum Hamiltonians for systems of particles subject to a two-body `\emph{point}' (or `\emph{contact}', or `\emph{zero range}') interaction between all pairs of the collection $\mathscr{C}$, since each such extension acts by construction as the free Hamiltonian on wave-functions supported away from the considered coincidence hyperplanes. (The self-adjoint Laplacian is the trivial extension and describes a system of $N$ non-interacting particles.)

In the concrete case $N=2$, $\mathscr{C}=\{(1,2)\}$, singular perturbations model a two-body system with point interaction; equivalently, factoring out the centre of mass and working in the Hilbert space $L^2(\mathbb{R}^d,\ud x)$ of the relative variable $x\equiv x_1-x_2$, they describe the motion of one particle subject to a point interaction supported at the origin. It is well-known \cite{albeverio-solvable} that the self-adjoint extensions in $L^2(\mathbb{R}^d)$ of the restriction of $\Delta$ to $C^\infty_0(\mathbb{R}^d\!\setminus\!\{0\})$ form a four-parameter family when $d=1$, a one-parameter family when $d=2,3$, and the trivial family, consisting only of the self-adjoint Laplacian, when $d\geqslant 4$.

We are going to focus more concretely on the three-dimensional setting: $d=3$. As opposite to $N=2$, where upon separating the centre of mass the deficiency indices are $(1,1)$, when $N\geqslant 3$ the operator $\Delta|_{C^\infty_0(\mathbb{R}^{3N}\setminus\Gamma_\mathscr{C})}$ even after the separation of the centre of mass has \emph{infinite} deficiency indices, and hence a much wider variety of extensions. Those of most stringent physical relevance are the singular perturbations of Ter-Martirosyan--Skornyakov type, named after the so-called Ter-Martirosyan--Skornyakov condition (known also as the Bethe-Peierls contact condition). Roughly speaking, this is an asymptotic condition that is required on the functions of the domain of the adjoint operator $(\Delta|_{C^\infty_0(\mathbb{R}^{3N}\setminus\Gamma_\mathscr{C})})^*$ when $|x_i-x_j|\to 0$, for each pair $(i,j)\in\mathscr{C}$ of particles subject to a contact interaction, in order to be physically meaningful wave-functions and thus to constitute a domain of (essential) self-adjointness for certain extensions of $\Delta|_{C^\infty_0(\mathbb{R}^{3N}\setminus\Gamma_\mathscr{C})}$. The form of the asymptotics is dictated by physical heuristics on the eigenvalue problem for the two-body Schr\"{o}dinger equation at low energy and with a potential supported on a very short range, and it depends on one real parameter for each two-body channel with interaction, essentially the two-body scattering length in that channel.

We refer to our recent work \cite{MO-2016} for a comprehensive discussion on point interactions realised as Ter-Martirosyan--Skornyakov Hamiltonians, including a survey of the previous vast literature. Intuitively, they are precisely those models that correspond to formal Hamiltonians
\begin{equation}
H\;=\;\sum_{j=1}^N(-{\textstyle\frac{1}{2m_j}}\Delta_{x_j})+\textrm{``}\!\!\sum_{(j,k)\in\mathscr{C}}\alpha_{ij}\,\delta(x_j-x_k)\textrm{\,''}\,.
\end{equation}

A well-known complication when $N\geqslant 3$, and $\mathscr{C}$ contains at least two pairs (that is, there is a point interaction in at least two distinct two-body channels), is that the restriction of the domain of $(\Delta|_{C^\infty_0(\mathbb{R}^{3N}\setminus\Gamma_\mathscr{C})})^*$ to those functions that satisfy the ``physical'' Ter-Martirosyan--Skornyakov condition may actually select a symmetric extension of $\Delta|_{C^\infty_0(\mathbb{R}^{3N}\setminus\Gamma_\mathscr{C})}$ that is \emph{not} essentially self-adjoint and in turn admits non-trivial self-adjoint extensions. (As stressed in \cite{MO-2016}, when instead $N=2$, then imposing a Ter-Martirosyan--Skornyakov condition does select a domain of self-adjointness.) This occurrence depends both on the partial permutation symmetry, of bosonic or fermionic type, which may be additionally required in the $N$-body Hilbert space, and on the mass parameters $m_1,\dots,m_N$ entering the definition of the multi-particle free Hamiltonian $-\Delta$.

Thus, in certain regimes of masses and of symmetries the sole specification, through Ter-Martirosyan--Skornya\-kov asymptotics, of the two-body scattering length of each two-body channels where a contact interaction is present is \emph{not enough} to identify unambiguously a well-posed multi-particle Hamiltonian, and a further prescription is needed.

This problem has been since long at the centre of many mathematical investigations, which we shall cite later in the course of our discussion, above all for systems of $N=N_1+N_2$ particles, $N_1$ identical bosons or fermions of one type and $N_2$ identical bosons or fermions of another type, with two-body zero interactions of zero range.

Within this general framework, in this work we continue our analysis, started in \cite{MO-2016}, of the prototypical $2+1$ fermionic case, the system of two identical fermions coupled by a contact interaction with a third particle of different nature.

In \cite{MO-2016} we surveyed the rich mathematical literature on this model and above all we proved all the rigorous steps through which a self-adjoint Hamiltonian of Ter-Martirosyan--Skornya\-kov type can be constructed based on the self-adjoint extension scheme of Kre{\u\i}n, Vi\v{s}ik, and Birman. We showed that those self-adjoint extensions of the `away-from-hyperplanes' free Hamiltonian, whose domain consists of functions satisfying the Ter-Martirosyan--Skornya\-kov asymptotics, are in one-to-one correspondence with the self-adjoint extensions of an explicit auxiliary operator $\mathcal{A}_{\lambda,\alpha}$ acting on the Hilbert space $H^{-1/2}(\mathbb{R}^3)$ equipped with a twisted, but equivalent scalar product -- the `\emph{space of charges}', following a successful nomenclature imported by an analogy with electrostatics \cite{dft-Nparticles-delta}. 
(Here $\alpha\in\mathbb{R}$ is the parameter imposed by the Ter-Martirosyan--Skornyakov asymptotics in each two-body channel and $\lambda>0$ is a sufficiently large constant.)
The identification of the correct auxiliary operator and of the correct space of charges was a fundamental point in our work \cite{MO-2016}, as opposite to an unfortunate and long-lasting misinterpretation of the previous literature, where instead it was believed that the Ter-Martirosyan--Skornya\-kov extensions were parametrised by suitable (and extensively studied) self-adjoint operators on $L^2(\mathbb{R}^3)$.

In turn, once the appropriate $\mathcal{A}_{\lambda,\alpha}$ is identified, one can investigate the possible occurrence of a multiplicity of self-adjoint Ter-Martirosyan--Skornya\-kov Hamiltonians, which is in fact the main goal of the present work, and for each such realisations one can study stability and spectral properties -- the stability of a distinguished realisation was recently proved in \cite{CDFMT-2012}, and later in \cite{Moser-Seiringer}, through the associated quadratic form.

More in detail, here are the results of our investigation. Our first main result is to make the auxiliary operator $\mathcal{A}_{\lambda,\alpha}$ acting on the space of charges completely explicit, as compared to the somewhat implicit characterisation we gave in \cite{MO-2016}.

We then make the structure of its self-adjoint extensions explicit, by combining symmetry arguments with general properties of the Kre{\u\i}n-Vi\v{s}ik-Birman theory. We classify such extensions through an extension formula based on the Friedrichs extension of $\mathcal{A}_{\lambda,\alpha}$ (it too being now determined explicitly) and additional singular charges $\Xi\in H^{-1/2}(\mathbb{R}^3)$ that have angular symmetry $\ell\geqslant 1$ and satisfy
\begin{equation}\label{eq:AstarXi=0_intro}
2\pi^2\sqrt{{\textstyle\frac{m(m+2)}{\:(m+1)^2}}\, p^2+\lambda}\;\widehat{\Xi}(p)+\int_{\mathbb{R}^3}\frac{\widehat{\Xi}(q)}{p^2+q^2+{\textstyle\frac{2}{m+1}}\, p\cdot q+\lambda}\,\ud q+\alpha\,\widehat{\Xi}(p)\;=\;0\,.
\end{equation}

Motivated by formal results in the physical literature, and by the findings of an alternative mathematical approach through quadratic forms, we then formulate a conjecture 
on the fact that there are two universal masses $m^{**}>m^{*}>0$ such that, for $m>m^*$ (which is the threshold for stable interactions, see e.g. \cite{CDFMT-2012})
the dimension of the space of solutions to \eqref{eq:AstarXi=0_intro} is always 0 when $\ell\neq 1$, it is 0 when $\ell=1$ and $m\geqslant m^{**}$, and it is instead 3 when $\ell=1$ and $m^*<m<m^{**}$. As we will show, $m^{**}$ represents the threshold mass above which one has a unique self-adjoint realisation of the Ter-Martirosyan--Skornyakov Hamiltonian with chosen mass $m$.

Our next result, assuming the validity of our conjecture, and based on the Kre{\u\i}n-Vi\v{s}ik-Birman extension scheme, is to reconstruct completely and for the first time the whole family of self-adjoint Hamiltonians of Ter-Martirosyan--Skornyakov type for the 2+1 fermionic model. This provides a \emph{complete explanation of the emergence of a multiplicity of self-adjoint realisations}.

We also show that through this operator-theoretic approach we reproduce previous partial constructions obtained by means of a quadratic form approach.

Remarkably, previous studies of the multiplicity of extensions, carried on with the same operator theoretic approach, had confused the correct space of charges and studied instead the solutions to \eqref{eq:AstarXi=0_intro} in the much smaller space $L^2(\mathbb{R}^3)$, instead of $H^{-1/2}(\mathbb{R}^3)$. This resulted in a threshold for the coexistence of a multiplicity of extensions which was lower than the value $m^{**}$ identified in the quadratic form approach, in the physical literature, and in our conjecture.

It then remains to prove our conjecture, which we find a highly non-trivial task, as reasonable and consistent such a conjecture appears to be. This is a goal to which we aim to devote our next investigation.

\section{Model set-up}

After removing the centre of mass, the free Hamiltonian of a three-dimensional system of two identical fermions of unit mass in relative positions $x_1,x_2$ with respect to a third particle of different species and with mass $m$ is the operator $-\Delta_{x_1}-\Delta_{x_2}-\frac{2}{m+1}\nabla_{x_1}\cdot\nabla_{x_2}$ acting on the Hilbert space
\begin{equation}
\cH\;=\;L^2_\mathrm{f}(\mathbb{R}^3\times\mathbb{R}^3,\ud x_1\ud x_2)\,,
\end{equation}
the subscript `f' standing for the fermionic sector of the $L^2$-space, i.e, the square-integrable functions that are anti-symmetric under exchange $x_1\leftrightarrow x_2$. The natural starting point is then the operator
\begin{equation}\label{eq:Hring_2+1}
\begin{split}
\mathring{H}\;&:=\;-\Delta_{x_1}-\Delta_{x_2}-\frac{2}{m+1}\nabla_{x_1}\cdot\nabla_{x_2} \\
\mathcal{D}(\mathring{H})\;&:=\;H^2_0((\mathbb{R}^3\times\mathbb{R}^3)\!\setminus\!(\Gamma_1\cup\Gamma_2))\cap\cH
\end{split}
\end{equation}
on $\cH$,
where 
\begin{equation}
\Gamma_j\;:=\;\{(x_1,x_2)\in\mathbb{R}^3\times\mathbb{R}^3\,|\,x_j=0\}\,,\qquad j=1,2,
\end{equation}
and 
\begin{equation}
H^2_0((\mathbb{R}^3\times\mathbb{R}^3)\!\setminus\!(\Gamma_1\cup\Gamma_2))\;=\;\overline{C^\infty_0((\mathbb{R}^3\times\mathbb{R}^3)\!\setminus\!(\Gamma_1\cup\Gamma_2))}^{\|\,\|_{H^2}}\,.
\end{equation}

$\mathring{H}$ is densely defined, closed, positive, and symmetric on $\cH$.\ As such, $\mathring{H}$ has equal deficiency indices (which are infinite, as stated in Proposition \ref{prop:D_Hdostar_2+1} below) and thus admits self-adjoint extensions, among which the Friedrichs extension $\mathring{H}_F$ is nothing but the self-adjoint negative Laplacian on $\cH$ with domain $\cH\cap H^2(\mathbb{R}^3\times\mathbb{R}^3,\ud x_1\ud x_2)$. Any other self-adjoint extension of $\mathring{H}$ has a natural interpretation of Hamiltonian of point interaction between each fermion and the third particle.

The following facts are known concerning the adjoint of $\mathring{H}$.

\begin{theorem}{Proposition}\label{prop:D_Hdostar_2+1}\emph{(\cite[Lemma 3 and Proposition 2]{MO-2016}.)}
Let $\lambda>0$. For $\xi\in H^{-1/2}(\mathbb{R}^3)$ define
\begin{equation}\label{eq:u_xi}
 \widehat{u}_\xi(p_1,p_2)\;:=\;\frac{\widehat{\xi}(p_1)-\widehat{\xi}(p_2)}{p_1^2+p_2^2+\mu\,p_1\cdot p_2+\lambda}\,,\qquad p_1,p_2\in\mathbb{R}^3\,,
\end{equation}
with
 \begin{equation}\label{eq:mu}
  \mu\;:=\;\frac{2}{m+1}\,.
 \end{equation}
\begin{itemize}
 \item[(i)] One has
 \begin{equation}\label{eq:ker_hring*_2+1}
 \begin{split}
 \!\!\ker (\mathring{H}^*+\lambda\mathbbm{1})\;&=\;\left\{ u_\xi\in L^2_\mathrm{f}(\mathbb{R}^3\!\times\!\mathbb{R}^3)\left|\!
 \begin{array}{c}
 \widehat{u}_\xi(p_1,p_2)\textrm{ given by }\eqref{eq:u_xi} \\
 \xi\in H^{-1/2}(\mathbb{R}^3)
 \end{array}\!\!\right.\right\}\,.
 \end{split}
 \end{equation}
 \item[(ii)] There exist constants $c_1,c_2>0$ such that for a generic $u_\xi\in\ker (\mathring{H}^*+\lambda\mathbbm{1})$ one has
 \begin{equation}\label{eq:uxi-equivalent-norms}
  c_1\|\xi\|_{H^{-1/2}(\mathbb{R}^3)}\leqslant\;\|u_\xi\|_{\cH}\;\leqslant c_2\|\xi\|_{H^{-1/2}(\mathbb{R}^3)}\,.
 \end{equation}
 \item[(iii)] The domain and the action of the adjoint of $\mathring{H}$ are given by
 \begin{equation}\label{eq:decompositionD_Hdostar_2+1}
  \mathcal{D}(\mathring{H}^*)\;=\;\left\{\!\!
  \begin{array}{c}
  g\in L^2_\mathrm{f}(\mathbb{R}^3\!\times\!\mathbb{R}^3)\quad\textrm{such that} \\
  \widehat{g}(p_1,p_2)=\displaystyle\widehat{f}(p_1,p_2)+\frac{\widehat{u}_\eta(p_1,p_2)}{p_1^2+p_2^2+\mu\,p_1\cdot p_2+\lambda}+\widehat{u}_\xi(p_1,p_2) \\
  \textrm{for}\quad f\in\mathcal{D}(\mathring{H})\,,\quad \eta,\xi\in H^{-1/2}(\mathbb{R}^3)
  \end{array}\!\!\right\} 
 \end{equation}
 and
 \begin{eqnarray}
  \!\!\!\!\!\!\!\!\!\!\!\!(\widehat{(\mathring{H}^*+\lambda) g)}\,(p_1,p_2)\!\!\!&=&\!\!\! (p_1^2+p_2^2+\mu\,p_1\cdot p_2+\lambda) \,\widehat{F}_\lambda(p_1,p_2) \label{eq:actionHdotstar_to_g-f_2+1}\,, \\
  \widehat{(\mathring{H}^* g)}(p_1,p_2)\!\!\!&=&\!\!\!(p_1^2+p_2^2+\mu\,p_1\cdot p_2) \widehat{g}(p_1,p_2)-(\widehat{\xi}(p_1)-\widehat{\xi}(p_2)), \label{eq:actionHdotstar_to_g-f_2+1bis}
 \end{eqnarray}
 where $u_\eta$ and $u_\xi$ are defined as in \eqref{eq:u_xi} above, and
\begin{equation}
\widehat{F}_\lambda(p_1,p_2)\;:=\;\widehat{f}(p_1,p_2)+\frac{\widehat{u}_\eta(p_1,p_2)}{p_1^2+p_2^2+\mu\,p_1\cdot p_2+\lambda}\,.
\end{equation}
\end{itemize}
\end{theorem}

As a consequence, it can be shown that (almost) all functions in $\mathcal{D}(\mathring{H}^*)$ satisfy the following large-momentum asymptotics.

\begin{theorem}{Proposition}\label{prop:asymptotic_integral_2+1}\emph{(\cite[Lemma 4]{MO-2016}.)}
Let $g$ be an arbitrary function in $\mathcal{D}(\mathring{H}^*)$. For a fixed $\lambda>0$ consider the decomposition \eqref{eq:decompositionD_Hdostar_2+1} of  $\widehat{g}$ in terms of $\widehat{f}$, $\widehat{u}_\xi$, $\widehat{u}_\eta$. 
%Proposition \ref{prop:D_Hdostar_2+1} %\eqref{eq:decompositionD_Hdostar} of c
%for some $f\in\mathcal{D}(\mathring{H})$ and some $\eta,\xi\in\mathbb{C}$.
Then, as $R\to +\infty$,
\begin{equation}\label{eq:g_*_asymptotics_2+1}
\begin{split}
\int_{\substack{ \\ \,p_2\in\mathbb{R}^3 \\ \! |p_2|<R}}{\:\widehat{g}(p_1,p_2) \,\ud p_2}\;&=\;4\pi\widehat{\xi}(p_1) R+\big(-\widehat{(T_\lambda\,\xi)}(p_1)+\!\!\begin{array}{c}\frac{1}{2}\end{array}\!\!\!\widehat{(W_\lambda\,\eta)}(p_1)\big)+o(1)% \\
%& \qquad\qquad\textrm{as}\qquad R\to +\infty\,.
\end{split}
\end{equation}
as a point-wise identity for almost every $p_1$, where
\begin{equation}\label{eq:Tlambda}
\widehat{(T_\lambda\,\xi)}(p)\;:=\;2\pi^2\sqrt{\nu p^2+\lambda}\;\widehat{\xi}(p)+\int_{\mathbb{R}^3}\frac{\widehat{\xi}(q)}{p^2+q^2+\mu p\cdot q+\lambda}\,\ud q\,,
\end{equation}
\begin{equation}\label{eq:Wlambda}
\widehat{(W_\lambda\,\xi)}(p)\;:=\;\frac{2\pi^2}{\sqrt{\nu p^2+\lambda}\,}\,\widehat{\xi}(p)-2\!\int_{\mathbb{R}^3}\frac{\widehat{\xi}(q)}{(p^2+q^2+\mu p\cdot q+\lambda)^2}\,\ud q\,,
\end{equation}
\begin{equation}\label{eq:nu}
\nu\;:=\;1-\frac{\:\mu^2}{4}\;=\;\frac{m(m+2)}{(m+1)^2}\,,
\end{equation}
and $\mu$ is given by \eqref{eq:mu}.
\end{theorem}

\begin{remark}{Remark}
Clearly,
\begin{equation}\label{eq:lambda-equiv-1}
(1+p_1^2+p_2^2)\;\sim\;(p_1^2+p_2^2+\mu p_1\cdot p_2+\lambda)
\end{equation}
(in the sense that each quantity controls the other from above and from below), because $\mu\in(0,2)$ (owing to \eqref{eq:mu} with $m> 0$) and $\lambda>0$. Since, for arbitrary $\varepsilon>0$ and $\xi\in H^{-\frac{1}{2}+\varepsilon}(\mathbb{R}^3)$,
\[
\begin{split}
\Big|\int_{\mathbb{R}^3}&\frac{\widehat{\xi}(q)}{p^2+q^2+\mu p\cdot q+\lambda}\,\ud q\Big|\;\leqslant\; \\
&\leqslant\;\|\xi\|_{H^{-\frac{1}{2}+\varepsilon}}\Big(\int_{\mathbb{R}^3}\frac{(q^2+1)^{\frac{1}{2}-\varepsilon}}{(p^2+q^2+\mu p\cdot q+\lambda)^2}\,\ud q\Big)^{1/2}\;<\;+\infty
\end{split}
\]
(owing to a Schwartz inequality in the first step and \eqref{eq:lambda-equiv-1} in the second one), we see that the integral in \eqref{eq:Tlambda} is finite for any $\xi\in H^{-\frac{1}{2}+\varepsilon}(\mathbb{R}^3)$, $\varepsilon>0$, while in general it diverges when $\xi\in H^{-\frac{1}{2}}(\mathbb{R}^3)$, as the example $\widehat{\xi}_0(q)\;:=\;\mathbf{1}_{\{|q|\geqslant 2\}}(|q|\ln|q|)^{-1}$ shows. A similar argument shows that the integral in \eqref{eq:Wlambda} is finite too at least for $\xi\in H^{-\frac{1}{2}}(\mathbb{R}^3)$. Summarising, $(\widehat{T_\lambda\,\xi})(p)$ is well-defined point-wise for almost every $p\in\mathbb{R}^3$ for $\xi\in H^{-\frac{1}{2}+\varepsilon}(\mathbb{R}^3)$, $\varepsilon>0$, whereas $(\widehat{W_\lambda\,\xi})(p)$ is so (at least) for $\xi\in H^{-1/2}(\mathbb{R}^3)$. Therefore, for a generic $g\in \mathcal{D}(\mathring{H}^*)$, and correspondingly for a generic charge $\xi\in H^{-\frac{1}{2}}(\mathbb{R}^3)$, the quantity in the l.h.s.~of \eqref{eq:g_*_asymptotics_2+1} is \emph{infinite} for every finite $R$ because the quantity $(\widehat{T_\lambda\,\xi})(p)$ is in general infinite when $\xi\in H^{-\frac{1}{2}}(\mathbb{R}^3)$. Instead, when additionally $\xi\in H^{-\frac{1}{2}+\varepsilon}(\mathbb{R}^3)$, with $\varepsilon>0$, the r.h.s.~of \eqref{eq:g_*_asymptotics_2+1} is finite (for almost every $p_1\in\mathbb{R}^3$): this case corresponds to a dense set of $g$'s in $\mathcal{D}(\mathring{H}^*)$, and for such $g$'s the quantity in the l.h.s.~of \eqref{eq:g_*_asymptotics_2+1} is \emph{finite} for finite $R$ and only diverges, linearly in $R$, as $R\to +\infty$.
\end{remark}

Each self-adjoint extension of $\mathring{H}$ is a self-adjoint restriction of $\mathring{H}^*$. By further imposing that in the asymptotics \eqref{eq:g_*_asymptotics_2+1} the constant-in-$R$ term be proportional to the coefficient of the linear-in-$R$ term, one selects a special class of restrictions of $\mathring{H}^*$ known as the Ter-Martirosyan--Skornyakov operators for this model. Explicitly, for $\alpha\in\mathbb{R}\cup\{\infty\}$ one considers the subspace $\mathcal{D}(\mathring{H}^{\textsc{tms}}_\alpha)\subset\mathcal{D}(\mathring{H}^*)$ consisting of all $g$'s such that
\begin{equation}\label{eq:g_tms_2+1}
\begin{split}
\int_{\substack{ \\ \,p_2\in\mathbb{R}^3 \\ \! |p_2|<R}}{\:\widehat{g}(p_1,p_2) \,\ud p_2}\;&=\;(4\pi R+\alpha)\,\widehat{\xi}(p_1)+o(1)\quad\textrm{as}\quad R\to +\infty\,,
\end{split}
\end{equation}
that is, those $g$'s whose charges $\eta$ and $\xi$ satisfy
\begin{equation}\label{eq:g_tms_2+1_part2}
\alpha\,\xi\;=\;-T_\lambda\xi +{\textstyle\frac{1}{2}}W_\lambda\eta\,,
\end{equation}
and then one defines
\begin{equation}
\mathring{H}^{\textsc{tms}}_\alpha\;:=\;\mathring{H}^*\upharpoonright\mathcal{D}(\mathring{H}^{\textsc{tms}}_\alpha)\,.
\end{equation}
So far \eqref{eq:g_tms_2+1} and \eqref{eq:g_tms_2+1_part2} are understood as point-wise identities. At least for $\mathring{H}^{\textsc{tms}}_\alpha$ operators defined by using sufficiently regular $\xi$'s and $\eta$'s it is not difficult to deduce from the general characterisation of $\mathring{H}^*$ given in Proposition \ref{prop:D_Hdostar_2+1} that
\begin{equation}\label{eq:TMS_as_symm_extns}
 \mathring{H}\;\subset\;\mathring{H}^{\textsc{tms}}_\alpha\;\subset\;(\mathring{H}^{\textsc{tms}}_\alpha)^*\;\subset\;\mathring{H}^*\,,
\end{equation}
that is, each $\mathring{H}^{\textsc{tms}}_\alpha$ is a \emph{symmetric} extension of $\mathring{H}$.

The choice \eqref{eq:g_tms_2+1}-\eqref{eq:g_tms_2+1_part2}, customarily referred to as the Ter-Martirosyan--Skornyakov condition, has a great physical relevance. The formal counterpart of \eqref{eq:g_tms_2+1} in position coordinates reads
\begin{equation}\label{eq:BPcontact}
g(x_1,x_2)\;=\;\widetilde{\xi}(x_1)\Big(\frac{1}{|x_2|}+\alpha\Big)+o(1)\qquad\textrm{ as}\quad x_2\to 0\,,
\end{equation}
and the short-scale asymptotics $(|x|^{-1}-(-\alpha^{-1})^{-1})$ in the relative coordinate $x$ is characteristic for wave-functions of a system of quantum particles subject to a two-body interaction of almost zero range and scattering length $(-\alpha)^{-1}$, a fact observed first by Bethe and Peierls in 1935 \cite{Bethe_Peierls-1935,Bethe_Peierls-1935-np} (whence the name of Bethe-Peierls contact condition for the asymptotics \eqref{eq:BPcontact}), and later exploited in momentum coordinates by Skornyakov and Ter-Martirosyan in 1956 \cite{TMS-1956}.

For this reason the interest towards point interaction Hamiltonians realised as self-adjoint extensions of $\mathring{H}$ is mainly focused on the physically relevant extensions of TMS type.

\section{Self-adjoint extension scheme for TMS operators}\label{sec:ext_scheme_for_H-TMS}

The primary issue concerning 2+1 fermionic TMS operators is their self-adjoint realisation, since this allows one to set up a well-posed and physically relevant model of point interaction for the 2+1 fermionic system.

In fact, as indicated in \eqref{eq:TMS_as_symm_extns}, for given $\alpha$ the TMS condition selects a priori only a symmetric extension $\mathring{H}^{\textsc{tms}}_\alpha$ of the away-from-hyperplanes free Hamiltonian $\mathring{H}$. For which values of the mass parameter $m$ is $\mathring{H}^{\textsc{tms}}_\alpha$ indeed self-adjoint or does it admit in turn self-adjoint extensions is the object of an active research activity for general $N+M$ fermionic models of point interactions.

Owing to the special structure of TMS operators, and ultimately to the positivity of the initial operator $\mathring{H}$, the issue of their self-adjointness is conveniently addressed to within the so-called Kre{\u\i}n-Vi\v{s}ik-Birman self-adjoint extension scheme for semi-bounded symmetric operators, a theory developed by  Kre{\u\i}n \cite{Krein-1947}, Vi\v{s}ik \cite{Vishik-1952}, and Birman \cite{Birman-1956} between the mid 1940's and the mid 1950's. For the present purposes we refer to the comprehensive discussion \cite{M-KVB2015}, as well as to the expository works \cite{Flamand-Cargese1965,Alonso-Simon-1980}.

The key point from the Kre{\u\i}n-Vi\v{s}ik-Birman theory is that the whole family of self-adjoint extensions of $\mathring{H}$ is one-to-one with self-adjoint operators on Hilbert subspaces of $\ker(\mathring{H}^*+\lambda\mathbbm{1})$, through an explicit, constructive correspondence between operators of the two classes. The goal is therefore to recognise the TMS condition \eqref{eq:g_tms_2+1_part2} as a self-adjointness condition for a suitable operator on $\ker(\mathring{H}^*+\lambda\mathbbm{1})$.

To this aim, the first step is to qualify the two maps $\xi\mapsto T_\lambda\xi$ and $\eta\mapsto W_\lambda\eta$ defined in \eqref{eq:Tlambda}-\eqref{eq:Wlambda} between convenient functional spaces. One has the following.

\begin{theorem}{Proposition}\label{prop:T-W} \emph{(\cite[Propositions 3 and 4, Corollary 2]{MO-2016} and Appendix \ref{appendix:bounds_ell_-1/2_3/2}.)}
\begin{itemize}
 \item[(i)] The expression \eqref{eq:Wlambda} defines a bounded, positive, and invertible linear operator $W_\lambda:H^{-1/2}(\mathbb{R}^3)\to H^{1/2}(\mathbb{R}^3)$, and for generic $u_\xi,u_\eta\in\ker (\mathring{H}^*+\lambda\mathbbm{1})$ one has
 \begin{equation}\label{eq:scalar_products}
 \langle u_\xi,u_\eta\rangle_{\cH}\;=\;\langle \xi,W_\lambda\eta\rangle_{H^{-\frac{1}{2}}(\mathbb{R}^3),H^{\frac{1}{2}}(\mathbb{R}^3)}\,.
 \end{equation}
 \item[(ii)] For each $s\geqslant 1$ the expression \eqref{eq:Tlambda} defines a densely defined and symmetric operator  $T_\lambda:\mathcal{D}(T_\lambda)\subset L^2(\mathbb{R}^3)\to L^2(\mathbb{R}^3)$ with domain $\mathcal{D}(T_\lambda):=H^s(\mathbb{R}^3)$.
 \item[(iv)] $T_\lambda$ commutes with the rotations in $\mathbb{R}^3$ and it is reduced with respect to the canonical decomposition
\begin{equation}\label{eq:L2_ell_decomposition}
L^2(\mathbb{R}^3)\;\cong\;\bigoplus_{\ell=0}^\infty L^2(\mathbb{R}^+,r^2\,\ud r)\otimes\mathrm{span}\{Y_{\ell,-\ell},\dots,Y_{\ell,\ell}\}\;\equiv\;\bigoplus_{\ell=0}^\infty L^2_\ell(\mathbb{R}^3)
\end{equation}
(where the $Y_{\ell,m}$'s are the spherical harmonics on $\mathbb{S}^2$). Thus, $T_\lambda$ leaves each $L^2_\ell(\mathbb{R}^3)$ invariant and decomposes as 
\begin{equation}\label{eq:T-Tell}
T_\lambda\;=\;\bigoplus_{\ell=0}^\infty \,T_\lambda^{(\ell)}\qquad T_\lambda^{(\ell)}\,\textrm{ densely defined and symmetric on } L^2_\ell(\mathbb{R}^3)\,.
\end{equation}
Each $T_\lambda^{(\ell)}$ acts non-trivially only on the first (radial) factor of $L^2_\ell(\mathbb{R}^3)\cong L^2(\mathbb{R}^+,r^2\,\ud r)\otimes\mathrm{span}\{Y_{\ell,-\ell},\dots,Y_{\ell,\ell}\}$, whereas it acts as the identity on the second (angular) factor.
\item[(v)] Analogously to \eqref{eq:L2_ell_decomposition}, let $H^s_\ell(\mathbb{R}^3)$ be the sector of $\ell$-th angular symmetry in $H^s(\mathbb{R}^3)$, that is,
\begin{equation}\label{eq:Hs_ell_decomposition}
\mathcal{F}\big(H^s_\ell(\mathbb{R}^3)\big)\;\cong\;L^2(\mathbb{R}^+,(1+\varrho^2)^s\varrho^2\ud \varrho)\otimes\mathrm{span}\{Y_{\ell,-\ell},\dots,Y_{\ell,\ell}\}\,,\quad\ell\in\mathbb{N}\,.
\end{equation}
%For each $\ell\geqslant 1$, and i
In terms of the notation of \eqref{eq:L2_ell_decomposition}-\eqref{eq:T-Tell}-\eqref{eq:Hs_ell_decomposition}, one has
\begin{equation}\label{eq:T-ell-3/2-1/2_excluded}
\|T_\lambda\xi\|_{H^{s-1}}\;\lesssim\;\|\xi\|_{H^{s}}\qquad s\in\,\textstyle{(-\frac{1}{2},\frac{3}{2})}
\end{equation}
that is, $T_\lambda$ maps continuously $H^s(\mathbb{R}^3)$ into $H^{s-1}(\mathbb{R}^3)$ for any $s\in(-\frac{1}{2},\frac{3}{2})$, whereas
\begin{equation}\label{eq:T-ell-3/2-1/2_included}
\|T_\lambda\xi\|_{H^{s-1}}\;\lesssim\;\|\xi\|_{H^{s}}\qquad\forall\xi\in H^{s}_\ell(\mathbb{R}^3) \qquad 
\begin{cases}
\;s\in\,\textstyle{[-\frac{1}{2},\frac{3}{2}]} \\
\quad \ell\geqslant 1\,,
\end{cases}
%\forall \xi\in H_{\ell}^{s}(\mathbb{R}^3)\textrm{ for }s\in\,\textstyle{[-\frac{1}{2},\frac{3}{2}]}\textrm{ and }\ell\geqslant 1\,,
\end{equation}
that is, $T_\lambda$ maps continuously $H^s_\ell(\mathbb{R}^3)$ into $H^{s-1}_\ell(\mathbb{R}^3)$ also for $s=-\frac{1}{2}$ and $s=\frac{3}{2}$, provided that $\ell\geqslant 1$.
% whence, in particular,
% \begin{equation}\label{eq:T-ell-3/2-1/2_part2}
%  T_\lambda^{(\ell)}\big(H^{s}_\ell(\mathbb{R}^3)\big)\;\subset\;H^{s-1}_\ell(\mathbb{R}^3)\qquad 
% \begin{cases}
% \;s\in\,\textstyle{[-\frac{1}{2},\frac{3}{2}]} \\
% \quad \ell\geqslant 1\,.
% \end{cases}
% \end{equation}
\item[(vi)] Instead, for $\ell=0$,
 \begin{equation}\label{eq:a_dense_not_in_it}
 T_\lambda\Big(\{\xi\,|\,\widehat{\xi}\in C^\infty_0(\mathbb{R}^3) \}\cap H^{3/2}_{\ell=0}(\mathbb{R}^3)\Big)\;\subset\;H^{\sigma}(\mathbb{R}^3)\quad\textrm{only for $\sigma<{\textstyle\frac{1}{2}}$}.
 \end{equation}
\end{itemize}
\end{theorem}

Owing to Proposition \ref{prop:T-W},
\begin{equation}\label{eq:W-scalar-product}
\langle \xi,\eta\rangle_{W_\lambda}\;:=\;\langle \xi,W_\lambda\,\eta\rangle_{H^{-\frac{1}{2}},H^{\frac{1}{2}}}\;=\;\langle u_\xi,u_\eta\rangle_{\cH}
\end{equation}
defines a scalar product in $H^{-\frac{1}{2}}(\mathbb{R}^3)$. It is \emph{equivalent} to the standard scalar product of $H^{-\frac{1}{2}}(\mathbb{R}^3)$, as follows by combining \eqref{eq:W-scalar-product} with \eqref{eq:uxi-equivalent-norms}.

We shall denote by $H^{-1/2}_{W_\lambda}(\mathbb{R}^3)$ the Hilbert space consisting of the $H^{-\frac{1}{2}}(\mathbb{R}^3)$-functions and equipped with the scalar product $\langle\cdot,\cdot\rangle_{W_\lambda}$.
Then the map
\begin{equation}\label{eq:isomorphism_Ulambda}
\begin{split}
U_\lambda\,:\,\ker (\mathring{H}^*+\lambda\mathbbm{1})\;&\;\xrightarrow[]{\;\;\;\cong\;\;\;}\;H^{-1/2}_{W_\lambda}(\mathbb{R}^3)\,,\qquad u_\xi \longmapsto \,\xi
\end{split}
\end{equation}
is an isomorphism between Hilbert spaces, with $\ker (\mathring{H}^*+\lambda\mathbbm{1})$ equipped with  the standard scalar product inherited from $\cH$.

One can therefore equivalently parametrise the self-adjoint extensions of $\mathring{H}$ in terms of self-adjoint operators acting on Hilbert subspaces of $\ker (\mathring{H}^*+\lambda\mathbbm{1})$ or of its unitarily equivalent version $H^{-1/2}_{W_\lambda}(\mathbb{R}^3)$. 

For our purposes, the relevant class of self-adjoint operators on $H^{-1/2}_{W_\lambda}(\mathbb{R}^3)$ to consider here are those of the form $\widetilde{\mathcal{A}}_{\lambda,\alpha}$ indicated in the Proposition \ref{lem:selfadj_hierarchy} below, and the reason of their relevance will be explained in the discussion that will follow right after.

\begin{theorem}{Proposition}\label{lem:selfadj_hierarchy}\emph{(\cite[Proposition 5]{MO-2016}.)}
The following data be given: 
\begin{itemize}
 \item two constants $\lambda>0$ and $\alpha\in\mathbb{R}$,
 \item for each integer $\ell\geqslant 1$, the densely defined and symmetric operator $T_\lambda^{(\ell)}$ on the Hilbert space $L^2_{\ell}(\mathbb{R}^3)$ with $\mathcal{D}(T_\lambda^{(\ell)}):=H^{3/2}_\ell(\mathbb{R}^3)$,
 \item for $\ell=0$, an operator $\mathcal{A}^{(\ell=0)}_{\lambda,\alpha}$ that is self-adjoint on the Hilbert space $H^{-1/2}_{W_\lambda,\ell=0}(\mathbb{R}^3)$ (a datum that will be further specified at a later stage).
 %and a densely defined and symmetric operator  $S_0$ on the Hilbert space $L^2_{\ell=0}(\mathbb{R}^3)$ with $\ran S_0\subset H^{1/2}(\mathbb{R}^3)$.
\end{itemize}
   With respect to the decomposition (see Remark \ref{rem:twisted_orthogonality} below)
\begin{equation}\label{eq:H-1/2_orthogonal_decomposition}
  H^{-1/2}_{W_\lambda}(\mathbb{R}^3)\;=\;\bigoplus_{\ell=0}^\infty \:H^{-1/2}_{W_\lambda,\ell}(\mathbb{R}^3)\,,
\end{equation}
   let
\begin{equation}\label{eq:def_A_all_sectors}
 \mathcal{A}_{\lambda,\alpha}\;:=\;\;\bigoplus_{\ell=0}^\infty \mathcal{A}^{(\ell)}_{\lambda,\alpha}\,,
\end{equation}
  where
\begin{equation}\label{eq:def_A_ell_geq1}
\begin{array}{rl}
 \mathcal{A}^{(\ell)}_{\lambda,\alpha}\!\!&:=\;2\,W_\lambda^{-1}(T^{(\ell)}_\lambda+\alpha\mathbbm{1}) \\
 \mathcal{D}(\mathcal{A}^{(\ell)}_{\lambda,\alpha})\!\!&:=\;\mathcal{D}(T^{(\ell)}_\lambda)\;=\;H^{3/2}_\ell(\mathbb{R}^3)
\end{array}\qquad\quad \ell\geqslant 1\,.
\end{equation}
 Then $\mathcal{A}_{\lambda,\alpha}$ is a densely defined and symmetric operator on  $H^{-1/2}_{W_\lambda}(\mathbb{R}^3)$. Moreover, if $\widetilde{\mathcal{A}}_{\lambda,\alpha}$ is a self-adjoint extension of $\mathcal{A}_{\lambda,\alpha}$ on $H^{-1/2}_{W_\lambda}(\mathbb{R}^3)$, then 
\begin{equation}\label{eq:unitary_equiv_A_Atilde_2+1}
A_{\lambda,\alpha}\;:=\;U_\lambda^{-1}\widetilde{\mathcal{A}}_{\lambda,\alpha} U_\lambda
\end{equation}
(where $U_\lambda$ is the isomorphism \eqref{eq:isomorphism_Ulambda}) is a self-adjoint operator on $\ker (\mathring{H}^*+\lambda\mathbbm{1})$. 
\end{theorem}

\begin{remark}{Remark}\label{rem:twisted_orthogonality}
With a slight abuse of notation, we use in \eqref{eq:H-1/2_orthogonal_decomposition} above, and in analogous formulas throughout, the same symbol of direct orthogonal sum without distinguishing between the usual Hilbert scalar product in $H^{-1/2}(\mathbb{R}^3)$ and the twisted scalar product in $H^{-1/2}_{W_\lambda}(\mathbb{R}^3)$. This is harmless because, owing to the definition \eqref{eq:W-scalar-product} for $\langle\cdot,\cdot\rangle_{W_\lambda}$ and the rotational symmetry of $W_\lambda$ defined in \eqref{eq:Wlambda}, elements belonging to subspaces of different symmetry $\ell$ are orthogonal in either scalar product.
\end{remark}

\begin{remark}{Remark}\label{rem:symmetry}
In short, $\mathcal{A}_{\lambda,\alpha}=2\,W_\lambda^{-1}(T_\lambda+\alpha\mathbbm{1})$ apart from a (self-adjoint) re-definition of this expression in the sector of symmetry $\ell=0$. This is to overcome the difficulty of defining $W_\lambda^{-1}T_\lambda$ when $\ell=0$, since owing to Proposition \ref{prop:T-W} the map $W_\lambda^{-1}$ cannot pull arbitrary functions $T_\lambda\xi$ back to $H^{-1/2}(\mathbb{R}^3)$. To our understanding, this simple fact had been overlooked in all the previous operator-theoretic approaches to the 2+1 fermionic model of TMS type \cite{Minlos-Shermatov-1989,Menlikov-Minlos-1991,Menlikov-Minlos-1991-bis,Shermatov-2003,Minlos-2011-preprint_May_2010,Minlos-2010-bis,Minlos-2012-preprint_30sett2011,Minlos-2012-preprint_1nov2012,Minlos-RusMathSurv-2014} until when we pointed it out in \cite{MO-2016}. We shall characterise $\mathcal{A}_{\lambda,\alpha}$ on the sector $\ell=0$ in Section \ref{sec:A_ell_0}
\end{remark}

The relevance, for the self-adjointness problem of TMS operators, of self-adjoint operators of the form $\widetilde{\mathcal{A}}_{\lambda,\alpha}$ on $H^{-1/2}_{W_\lambda}(\mathbb{R}^3)$, is dictated by the KVB extension theory. Indeed, such $\widetilde{\mathcal{A}}_{\lambda,\alpha}$'s turn out to be precisely the `\emph{extension parameters}', in the sense of the Kre{\u\i}n-Vi\v{s}ik-Birman theory, for the family of self-adjoint extensions of $\mathring{H}$ on $\cH$ of TMS type.

More precisely, as stated in Proposition \ref{lem:selfadj_hierarchy} above, each $A_{\lambda,\alpha}=U_\lambda^{-1}\widetilde{\mathcal{A}}_{\lambda,\alpha} U_\lambda$ is densely defined and self-adjoint on $\ker(\mathring{H}^*+\lambda\mathbbm{1})$ and as such it qualifies a self-adjoint extension of $\mathring{H}$. Let us denote by $H_{\alpha}$ such extension. The correspondence $A_{\lambda,\alpha}\leftrightarrow H_{\alpha}+\lambda\mathbbm{1}$ is prescribed by the Kre{\u\i}n-Vi\v{s}ik-Birman theory and reads (\cite[Theorem 3.4]{M-KVB2015})
\begin{equation}\label{eq:D_Halpha_2+1}
\begin{split}
\mathcal{D}(H_{\alpha})\;&:=\;\left\{g=f+(\mathring{H}_F+\lambda\mathbbm{1})^{-1}(A_{\lambda,\alpha}u_\xi)+u_\xi\left|\!
\begin{array}{c}
f\in\mathcal{D}(\mathring{H}) \\
u_\xi\in\mathcal{D}(A_{\lambda,\alpha})
\end{array}\!\!\!\right.\right\}\,, \\
H_{\alpha}\;&:=\;\mathring{H}^*\upharpoonright\mathcal{D}(\mathring{H}_{\alpha})\,.
\end{split}
\end{equation}
(Thus, $H_{\alpha}$ is a maximally defined extension in the sense that its extension parameter $A_{\lambda,\alpha}$ is densely defined on the deficiency space $\ker(\mathring{H}^*+\lambda\mathbbm{1})$; there are also self-adjoint extensions of $\mathring{H}$ whose extension parameter is self-adjoint on a closed \emph{proper} subspace of $\ker(\mathring{H}^*+\lambda\mathbbm{1})$, however such extensions do not find room in the present discussion, as they are not of TMS type.)

Direct comparison between \eqref{eq:D_Halpha_2+1} and \eqref{eq:decompositionD_Hdostar_2+1} shows that $\mathcal{D}(H_{\alpha})$ is obtained as a restriction $\mathcal{D}(\mathring{H}^*)$ by imposing the condition
\begin{equation}\label{eq:TMS_2+1_a}
u_\eta\;=\;A_{\lambda,\alpha} u_\xi
\end{equation}
as an identity in $\ker (\mathring{H}^*+\lambda\mathbbm{1})$, which, by the unitary equivalence \eqref{eq:unitary_equiv_A_Atilde_2+1}, is tantamount as 
\begin{equation}\label{eq:TMS_2+1_b}
\eta\;=\;\widetilde{\mathcal{A}}_{\lambda,\alpha}\,\xi
\end{equation}
as an identity in $H^{-1/2}_{W_\lambda}(\mathbb{R}^3)$. If we restrict to $\ell\geqslant 1$ and $\xi$ is taken in $\mathcal{D}(\mathcal{A}^{(\ell)}_{\lambda,\alpha})$ instead of $\mathcal{D}(\widetilde{\mathcal{A}}^{(\ell)}_{\lambda,\alpha})$ (we make this choice here only for presentational purposes: in Sections \ref{sec:reconstr_Halpha} and \ref{sec:emergence_of_TMS} we shall release this restriction), then \eqref{eq:TMS_2+1_b} reads
\begin{equation}\label{eq:TMS_2+1_c}
\alpha\,\xi\;=\;-T^{(\ell)}_\lambda\xi +\!\!\!\begin{array}{c}\frac{1}{2}\end{array}\!\!\!W_\lambda\eta\qquad\qquad\forall\xi\in H^{3/2}_\ell(\mathbb{R}^3)\,,\quad\ell\geqslant 1\,,
\end{equation}
as an identity in $H^{1/2}_{\ell}(\mathbb{R}^3)$, owing to the definition \eqref{eq:def_A_ell_geq1}. Plugging \eqref{eq:TMS_2+1_c} into \eqref{eq:g_*_asymptotics_2+1} yields the following asymptotics for elements in $\mathcal{D}(\mathring{H}_{\alpha})$ as $R\to\infty$:
\begin{equation}\label{eq:TMS_cond_asymptotics_2+1}
\begin{split}
\int_{\substack{ \\ \,p_2\in\mathbb{R}^3 \\ \! |p_2|<R}}{\:\widehat{g}(p_1,p_2) \,\ud p_2}\;&=\;\widehat{\xi}(p_1)(4\pi R+\alpha)+o(1) \quad\quad (\,\xi\in H^{3/2}_\ell(\mathbb{R}^3)\,)\,.
%& \qquad\qquad\textrm{as}\qquad R\to +\infty\,.
\end{split}
\end{equation}

Let us summarise the above discussion as follows.

\begin{theorem}{Proposition}
 \begin{itemize}
 \item[(i)] Any of the two equivalent conditions \eqref{eq:TMS_2+1_a}, \eqref{eq:TMS_2+1_b} is a condition of self-adjointness for restrictions of $\mathring{H}^*$
 \item[(ii)] Each such condition selects the self-adjoint extension
   \begin{equation}\label{eq:ext_formula_etaAxi}
 H_{\alpha}\;=\;\mathring{H}^*\upharpoonright\big\{g\in\mathcal{D}(\mathring{H}^*)\,|\,\eta=\widetilde{\mathcal{A}}_{\lambda,\alpha}\,\xi\big\}
 \end{equation}
 given by the restriction of $\mathring{H}^*$ to those elements of $\mathcal{D}(\mathring{H}^*)$ whose charges $\xi$ and $\eta$, in terms of the decomposition \eqref{eq:decompositionD_Hdostar_2+1}, instead of belonging generically to $H^{-1/2}(\mathbb{R}^3)$ are such that $\xi$ belongs to the domain of $\widetilde{\mathcal{A}}_{\lambda,\alpha}$ and $\eta$ is of the form $\widetilde{\mathcal{A}}_{\lambda,\alpha}\xi$.
 \item[(iii)] Any $g\in\mathcal{D}(H_{\alpha})$ with charge $\xi\in\mathcal{D}(T_\lambda^{\ell})=H^{3/2}_\ell(\mathbb{R}^3)$, $\ell\geqslant 1$, satisfies the TMS asymptotics \eqref{eq:g_tms_2+1}.
\end{itemize}
\end{theorem}

Informally speaking: \emph{each self-adjoint extension $H_{\alpha}$ of $\mathring{H}$ is, on each sector of charges with symmetry $\ell\geqslant 1$, a Ter-Martirosyan--Skornyakov  operator}. (We shall supplement this picture in Sections \ref{sec:reconstr_Halpha} and \ref{sec:emergence_of_TMS} also with the case $\ell=0$.) In this sense
\begin{equation}\label{eq:TMS_as_selfadj_extns}
 \mathring{H}\;\subset\;\mathring{H}^{\textsc{tms}}_\alpha\;\subset\;H_{\alpha}=\;H_{\alpha}^*\;\subset\;(\mathring{H}^{\textsc{tms}}_\alpha)^*\;\subset\;\mathring{H}^*\qquad(\ell\geqslant 1)\,.
\end{equation}
%In a sense, instead, the sector $\ell=0$ is not designed to support simultaneously the self-adjointness \emph{and} the TMS condition. We shall discuss this aspect in detail in the following Sections.

Moreover, it is established that \emph{the issue of the self-adjoint realisation of TMS Hamiltonians on the physical Hilbert space $\cH$ is tantamount as the self-adjoint realisation of the auxiliary operator $\mathcal{A}_{\lambda,\alpha}$ on the Hilbert space of charges $H^{1/2}_{W_\lambda}(\mathbb{R}^3)$}.

\begin{remark}{Remark}
The intuition of re-phrasing the self-adjointness problem of TMS Hamiltonians in terms of the self-adjointness of an auxiliary operator acting on a smaller Hilbert space -- the `space of charges', in the spirit of the extension theory of Kre{\u\i}n, Vi\v{s}ik, and Birman -- is originally due to Minlos and Faddeev in a seminal short announcement published in 1962 \cite{Minlos-Faddeev-1961-1}. However, from the very beginning an unfortunate and long-lasting misinterpretation established, as it was thought that the self-adjointness problem for what we denoted here with $\mathcal{A}_{\lambda,\alpha}$ on $H^{-1/2}_{W_\lambda}(\mathbb{R}^3)$ is equivalent to the self-adjointness problem of $T_\lambda$ on $L^2(\mathbb{R}^3)$. Not only must this statement be well formulated in each symmetry sector, 
%for otherwise $W_\lambda^{-1}T_\lambda$ might not be defined (see Remark \ref{rem:symmetry} above), 
but most importantly the two problems are \emph{not} the same even when full care is taken of the symmetry. Indeed, as proved in \cite[Remark 10]{MO-2016}, if in addition to the assumptions of Proposition \ref{lem:selfadj_hierarchy} above one assumes also that the $T_\lambda^{(\ell)}$'s for $\ell\geqslant 1$ are defined on a larger domain so as to become \emph{self-adjoint} in $L^2_\ell(\mathbb{R}^3)$, then for the densely defined and symmetric operator $\mathcal{A}^{(\ell)}_{\lambda,\alpha}$ on  $H^{-1/2}_{W_\lambda,\ell}(\mathbb{R}^3)$ one would have
\begin{equation*}
\mathcal{D}(\mathcal{A}_{\lambda,\alpha}^{(\ell)\star})\cap L^2_\ell(\mathbb{R}^3)\;=\;\mathcal{D}(\mathcal{A}^{(\ell)}_{\lambda,\alpha})\qquad\qquad(T_\lambda^{(\ell)}=T_\lambda^{(\ell)*})
\end{equation*}
(here $\mathcal{A}_{\lambda,\alpha}^\star$ denotes the adjoint in $H^{-1/2}_{W_\lambda}(\mathbb{R}^3)$).
However, this is \emph{not} enough to claim that the self-adjointness of $T^{(\ell)}_\lambda$ implies the self-adjointness of $\mathcal{A}^{(\ell)}_{\lambda,\alpha}$: the latter could still have (and in general it does have) a larger adjoint and thus admit non-trivial self-adjoint extensions. To our understanding this point was systematically missed in the very work of Minlos and Faddeev (see the statement in \cite{Minlos-Faddeev-1961-1} right after equation (19) therein, where the analogue of our $\mathcal{A}_{\lambda,\alpha}$ is introduced), in the 1987 seminal work of Minlos on the three-particle system \cite{Minlos-1987}, and throughout the subsequent literature on the operator-theoretic approach to the 2+1 fermionic model of TMS type \cite{Minlos-Shermatov-1989,Menlikov-Minlos-1991,Menlikov-Minlos-1991-bis,Shermatov-2003,Minlos-2011-preprint_May_2010,Minlos-2010-bis,Minlos-2012-preprint_30sett2011,Minlos-2012-preprint_1nov2012,Minlos-RusMathSurv-2014}. The discrepancy between the self-adjointness problem of $T_\lambda$ on $L^2(\mathbb{R}^3)$ and of $\mathring{H}^{\textsc{tms}}_\alpha$ on $\cH$ was proved in a sharp quantitative manner 
for the first time in 2015 in the work \cite{CDFMT-2012} by one of us in collaboration with Correggi, Finco, Dell'Antonio, and Teta, using quadratic form methods and working in the case $\alpha=0$. Then, in our recent work \cite{MO-2016} we have finally provided a consistent explanation of the difference between such two problems using the extension theory in the operator-theoretic language.
\end{remark}

\section{Extension scheme for the auxiliary $\mathcal{A}_{\lambda,\alpha}$: case $\ell\geqslant 1$}\label{sec:ext_scheme_for_A}

It follows from Propositions \ref{prop:T-W} and \ref{lem:selfadj_hierarchy} that the self-adjointness problem of the operator $\mathcal{A}_{\lambda,\alpha}$ defined in \eqref{eq:def_A_all_sectors} and of its self-adjoint extensions on the space of charges $H^{-1/2}_{W_\lambda}(\mathbb{R}^3)$ is separated (reduced) for each value of the angular momentum $\ell\in\mathbb{N}_0$. Indeed, we have the following decomposition of $H^{-1/2}_{W_\lambda}(\mathbb{R}^3)$ and reduction of $\mathcal{A}_{\lambda,\alpha}$
\begin{equation*}
 H^{-1/2}_{W_\lambda}(\mathbb{R}^3)\;=\;\bigoplus_{\ell=0}^{\infty}\,H^{-1/2}_{W_\lambda,\ell}(\mathbb{R}^3)\,,\qquad \mathcal{A}_{\lambda,\alpha}\;=\;\bigoplus_{\ell=0}^{\infty}\,\mathcal{A}_{\lambda,\alpha}^{(\ell)}\,,
\end{equation*}
where each $H^{-1/2}_{W_\lambda,\ell}(\mathbb{R}^3)$ is the sector of $\ell$-th angular symmetry of $H^{-1/2}(\mathbb{R}^3)$, equipped with the scalar product $\langle\cdot,\cdot\rangle_{W_\lambda}$ \eqref{eq:W-scalar-product}, and it is a reducing subspace for $\mathcal{A}_{\lambda,\alpha}$.

In this Section we discuss the self-adjointness problem of $\mathcal{A}_{\lambda,\alpha}^{(\ell)}$ in the case $\ell\geqslant 1$. For convenience of presentation, we restrict our analysis to the the regime of masses $m>m^*$, where $m^*\approx(13.607)^{-1}$ is the unique root of $\Lambda(m)=1$ and 
\begin{equation}\label{eq:Lambdam}
\Lambda(m)\;:=\;{\textstyle\frac{2}{\pi}}(m+1)^2\Big(\frac{1}{\sqrt{m(m+2)}}-\arcsin\frac{1}{m+1}\Big)\,,
\end{equation}
a positive, smooth, monotone decreasing function often also referred to as the Efimov transcendental function. We see that, thanks to this choice, $\mathcal{A}_{\lambda,\alpha}^{(\ell)}$ is strictly positive.

\begin{theorem}{Proposition}\label{prop:Apos}
For $m>m^*$ and $\ell\geqslant 1$ the operator $\mathcal{A}_{\lambda,\alpha}^{(\ell)}$ is densely defined, symmetric, and semi-bounded from below on the space $H^{-1/2}_{W_\lambda,\ell}(\mathbb{R}^3)$, with strictly positive bottom  when $\alpha\geqslant 0$ and $\lambda>0$, or also when $\alpha<0$ provided that $\lambda$ is large enough.
\end{theorem}

\begin{proof}
The fact that $\mathcal{A}_{\lambda,\alpha}^{(\ell)}$ is densely defined and symmetric on  $H^{-1/2}_{W_\lambda,\ell}(\mathbb{R}^3)$ is already part of Proposition \ref{lem:selfadj_hierarchy}. Moreover, for arbitrary $\xi\in\mathcal{D}(\mathcal{A}_{\lambda,\alpha}^{(\ell)})\!=\mathcal{D}(T_\lambda^{(\ell)})=H^{3/2}_\ell(\mathbb{R}^3)$, owing to \eqref{eq:T-ell-3/2-1/2_included} and \eqref{eq:def_A_ell_geq1} one has
\[
\begin{split}
{\textstyle\frac{1}{2}}\langle\xi,\mathcal{A}_{\lambda,\alpha}^{(\ell)}\xi\rangle_{H^{-1/2}_{W_\lambda}}\;&=\;\langle\xi,(T_\lambda^{(\ell)}+\alpha\mathbbm{1})\xi\rangle_{H^{-\frac{1}{2}},H^{\frac{1}{2}}}\;=\;\int_{\mathbb{R}^3}\overline{\widehat{\xi}(p)}\:\widehat{(T_\lambda^{(\ell)}\xi)}(p)\,\ud p+\alpha \|\xi\|_2^2\,.
\end{split}
\]
In \cite[Proposition 3.1]{CDFMT-2012} it was proved that
\begin{equation}\label{eq:positivity_from_CDFMT2012}
\begin{split}
\int_{\mathbb{R}^3}\overline{\widehat{\xi}(p)}\:\widehat{(T_\lambda^{(\ell)}\xi)}(p)\,\ud p\;&\geqslant\;2\pi^2(1-\Lambda(m))\int_{\mathbb{R}^3}\sqrt{\nu p^2+\lambda}\,|\widehat{\xi}(p)|^2\,\ud p\,. %\\
%&\geqslant\;2\pi^2(1-\Lambda(m))\,\|\xi\|_2^2
\end{split}
\end{equation}
Therefore, recalling that $\Lambda(m)<1$ for $m>m^*$,
\[
\begin{split}
\langle\xi,\mathcal{A}_{\lambda,\alpha}^{(\ell)}\xi\rangle_{H^{-1/2}_{W_\lambda}}\;&\geqslant\;(4\pi^2\sqrt{\lambda}\,(1-\Lambda(m))+2\alpha)\,\|\xi\|_2^2\,.
\end{split}
\]
The constant in the r.h.s.~above is strictly positive when $\alpha\geqslant 0$ and $\lambda>0$, or also when $\alpha<0$ provided that $\lambda$ is large enough. Since (Proposition \ref{prop:D_Hdostar_2+1}(ii))
\[
\|\xi\|_2\;\geqslant\;\|\xi\|_{H^{-1/2}}\;\approx\;\|\xi\|_{H^{-1/2}_{W_\lambda}}\,,
\]
the conclusion then follows.
\end{proof}

As such, $\mathcal{A}_{\lambda,\alpha}^{(\ell)}$ may have other self-adjoint extensions on  $H^{-1/2}_{W_\lambda,\ell}(\mathbb{R}^3)$ than its Friedrichs extension, and it is possible to classify them by means of the Kre{\u\i}n-Vi\v{s}ik-Birman extension theory, more precisely, the self-adjoint extensions of $\mathcal{A}_{\lambda,\alpha}^{(\ell)}$ are one-to-one with self-adjoint operators acting on Hilbert subspaces of $\ker(\mathcal{A}_{\lambda,\alpha}^{(\ell)})^\star$. Observe that $(\mathcal{A}_{\lambda,\alpha}^{(\ell)})^\star$ denotes the adjoint of $\mathcal{A}_{\lambda,\alpha}^{(\ell)}$ with respect to the Hilbert space $H^{-1/2}_{W_\lambda,\ell}(\mathbb{R}^3)$.

The two inputs for the extension formula given by the Kre{\u\i}n-Vi\v{s}ik-Birman theory (\cite[Theorem 3.4]{M-KVB2015}) are the Friedrichs extension $\mathcal{A}_{\lambda,\alpha}^{(\ell),F}$ of $\mathcal{A}_{\lambda,\alpha}^{(\ell)}$ (namely the unique self-adjoint extension whose domain is entirely contained in the form domain of $\mathcal{A}_{\lambda,\alpha}^{(1)}$ and hence the highest of all other extensions) and the space $\ker(\mathcal{A}_{\lambda,\alpha}^{(\ell)})^\star$. We characterise $\mathcal{A}_{\lambda,\alpha}^{(\ell),F}$ first.

\begin{theorem}{Proposition}\label{prop:Friedrichs_ext_of_A} Let $m>m^*$, $\ell\geqslant 1$, and $\lambda>0$ large enough (depending on $\alpha$) so as to make the bottom of $\mathcal{A}_{\lambda,\alpha}^{(\ell)}$ strictly positive, as found in Proposition \ref{prop:Apos}. The Friedrichs extension of $\mathcal{A}_{\lambda,\alpha}^{(\ell)}$ (with respect to the Hilbert space $H^{-1/2}_{W_\lambda,\ell}(\mathbb{R}^3)$) is the operator
\begin{equation}\label{eq:DFriedrA}
\begin{split}
\mathcal{D}(\mathcal{A}_{\lambda,\alpha}^{(\ell),F})\;&=\;\{\xi\in H^{1/2}_{\ell}(\mathbb{R}^3)\,|\,(T_\lambda^{(\ell)}+\alpha\mathbbm{1})\xi\in  H^{1/2}_{\ell}(\mathbb{R}^3)\} \\
\mathcal{A}_{\lambda,\alpha}^{(\ell),F}\,\xi\;&=\;2\,W_\lambda^{-1}(T_\lambda^{(\ell)}+\alpha\mathbbm{1})\,\xi\,.
\end{split}
\end{equation}
Clearly, $\mathcal{D}(\mathcal{A}_{\lambda,\alpha}^{(\ell),F})\supset H^{3/2}_{\ell}(\mathbb{R}^3)$.
The (Friedrichs) quadratic form of $\mathcal{A}_{\lambda,\alpha}^{(\ell),F}$ is given by
\begin{equation}\label{eq:DformA}
\begin{split}
\mathcal{D}[\mathcal{A}_{\lambda,\alpha}^{(\ell),F}]\;&=\;H^{1/2}_{\ell}(\mathbb{R}^3) \\
\mathcal{A}_{\lambda,\alpha}^{(\ell),F}[\xi_1,\xi_2]\;&=\;2\,\langle\,\xi_1,(T_\lambda^{(\ell)}+\alpha\mathbbm{1})\xi_2\rangle_{H^{\frac{1}{2}},H^{-\frac{1}{2}}}\,.
\end{split}
\end{equation}
(Observe that $\langle\xi_1,(T_\lambda^{(\ell)}+\alpha\mathbbm{1})\xi_2\rangle_{H^{\frac{1}{2}},H^{-\frac{1}{2}}}$ is \emph{finite} for $\xi_2\in H^{1/2}_{\ell}(\mathbb{R}^3)$, due to \eqref{eq:T-ell-3/2-1/2_included}.)
\end{theorem}

\begin{proof}
We start with the quadratic form associated with $\mathcal{A}_{\lambda,\alpha}^{(\ell),F}$ (after all, the Friedrichs extension is a form construction).
The form domain of the Friedrichs extension of $\mathcal{A}_{\lambda,\alpha}^{(\ell)}$ is given (see, e.g., \cite[Theorem A.2]{M-KVB2015} or \cite[Section 10.4]{schmu_unbdd_sa}) by the completion of $\mathcal{D}(\mathcal{A}_{\lambda,\alpha}^{(\ell)})=H^{3/2}_\ell(\mathbb{R}^3)$ in the norm 
\[
\|\xi\|_{\mathcal{A}_{\lambda,\alpha}^{(\ell)}}\;:=\;\Big(\langle\xi,\mathcal{A}_{\lambda,\alpha}^{(\ell)}\xi\rangle_{H^{-1/2}_{W_\lambda}}\Big)^{\!1/2}\;=\;\Big(2\,\langle\xi,(T_\lambda^{(\ell)}+\alpha\mathbbm{1})\xi\rangle_{H^{-\frac{1}{2}},H^{\frac{1}{2}}}\Big)^{\!1/2}\,.
\]
(It was crucial here to have taken $\lambda$ large enough so as to make the bottom of $\mathcal{A}_{\lambda,\alpha}^{(\ell)}$ strictly positive.)
We now use the fact, proved in \cite[Eq.(3.52)]{CDFMT-2012}, that for $m>m^*$ the expression $\langle\xi,(T_\lambda^{(\ell)}+\alpha\mathbbm{1})\xi\rangle$ defines a (square) norm \emph{equivalent to the $H^{1/2}$-norm} (what we denote here by $\langle\xi,(T_\lambda^{(\ell)}+\alpha\mathbbm{1})\xi\rangle$ is precisely the quantity $\Phi^\lambda_\alpha[\xi]$ in \cite{CDFMT-2012}). This yields the first identity in \eqref{eq:DformA}. Moreover (see, e.g., \cite[Theorem A.2]{M-KVB2015}), for $\xi_1,\xi_2\in\mathcal{D}[\mathcal{A}_{\lambda,\alpha}^{(\ell),F}]=H^{1/2}_{\ell}(\mathbb{R}^3)$ one has
\[
{\textstyle\frac{1}{2}}\mathcal{A}_{\lambda,\alpha}^{(\ell),F}[\xi_1,\xi_2]\;=\;{\textstyle\frac{1}{2}}\lim_{n\to\infty}\langle\xi_1^{(n)},\mathcal{A}_{\lambda,\alpha}^{(\ell)}\xi_2^{(n)}\rangle_{H^{-1/2}_{W_\lambda}}\;=\;\lim_{n\to\infty}\langle\xi_1^{(n)},(T_\lambda^{(\ell)}+\alpha\mathbbm{1})\xi_2^{(n)}\rangle_{H^{\frac{1}{2}},H^{-\frac{1}{2}}}\,,
\]
where $(\xi_1^{(n)})_n$ and $(\xi_2^{(n)})_n$ are sequences in $\mathcal{D}(\mathcal{A}_{\lambda,\alpha}^{(\ell)})=H^{3/2}_\ell(\mathbb{R}^3)$ such that $\xi_1^{(n)}\to\xi_1$ and $\xi_2^{(n)}\to\xi_2$ in the $\|\cdot\|_{\mathcal{A}_{\lambda,\alpha}^{(\ell)}}$-norm, namely the $H^{1/2}$-norm. Using again the norm equivalence proved in \cite[Eq.(3.52)]{CDFMT-2012}, one can pass to the limit inside the scalar product, thus obtaining the second identity in \eqref{eq:DformA}. Also the symmetry property
\[
\langle\xi_1^{(n)},(T_\lambda^{(\ell)}+\alpha\mathbbm{1})\xi_2^{(n)}\rangle_{H^{\frac{1}{2}},H^{-\frac{1}{2}}}\;=\;\langle(T_\lambda^{(\ell)}+\alpha\mathbbm{1})\xi_1^{(n)},\xi_2^{(n)}\rangle_{H^{-\frac{1}{2}},H^{\frac{1}{2}}}
\]
%$\langle\xi_1^{(n)},(T_\lambda^{(\ell)}+\alpha\mathbbm{1})\xi_2^{(n)}\rangle_{H^{\frac{1}{2}},H^{-\frac{1}{2}}}=\langle(T_\lambda^{(\ell)}+\alpha\mathbbm{1})\xi_1^{(n)},\xi_2^{(n)}\rangle_{H^{-\frac{1}{2}},H^{\frac{1}{2}}}$ 
is retained in the limit, whence indeed $\mathcal{A}_{\lambda,\alpha}^{(\ell),F}[\xi_1,\xi_2]=\overline{\mathcal{A}_{\lambda,\alpha}^{(\ell),F}[\xi_2,\xi_1]}$.

The \emph{operator} $\mathcal{A}_{\lambda,\alpha}^{(\ell),F}$ is then derived by its quadratic form in the usual manner. Its domain is the space
\[
\mathcal{D}(\mathcal{A}_{\lambda,\alpha}^{(\ell),F})\;=\;\left\{\xi\in \mathcal{D}[\mathcal{A}_{\lambda,\alpha}^{(\ell),F}]=H^{1/2}_{W_\lambda,\ell}(\mathbb{R}^3)\left|\begin{array}{c}
\exists\,\zeta_\xi\in H^{-1/2}_{W_\lambda,\ell}(\mathbb{R}^3)\textrm{ such that} \\
\langle\eta,\zeta_\xi\rangle_{H^{1/2}_{W_\lambda}}=\mathcal{A}_{\lambda,\alpha}^{(\ell),F})[\eta,\xi] \\
\forall\eta\in \mathcal{D}[\mathcal{A}_{\lambda,\alpha}^{(\ell),F}]=H^{1/2}_{W_\lambda,\ell}(\mathbb{R}^3)
\end{array}\!\right.\right\}
\]
and its action is then given by
\[
\mathcal{A}_{\lambda,\alpha}^{(\ell),F}\,\xi\;=\;\zeta_\xi\,.
\]
Since $\langle\eta,\zeta_\xi\rangle_{H^{1/2}_{W_\lambda}}=\langle\eta,W_\lambda\zeta_\xi\rangle$, $\mathcal{A}_{\lambda,\alpha}^{(\ell),F})[\eta,\xi]=2\,\langle\,\eta,(T_\lambda^{(\ell)}+\alpha\mathbbm{1})\xi\rangle_{H^{\frac{1}{2}},H^{-\frac{1}{2}}}$, and $W_\lambda$ is a bijection $H^{-\frac{1}{2}}(\mathbb{R}^3)\to H^{\frac{1}{2}}(\mathbb{R}^3)$, the condition $\langle\eta,\zeta_\xi\rangle_{H^{1/2}_{W_\lambda}}=\mathcal{A}_{\lambda,\alpha}^{(\ell),F})[\eta,\xi]$ implies that $(T_\lambda^{(\ell)}+\alpha\mathbbm{1})\xi\in H^{\frac{1}{2}}(\mathbb{R}^3)$ and that $\zeta_\xi=2\,W_\lambda^{-1}(T_\lambda^{(\ell)}+\alpha\mathbbm{1})\xi$, which proves \eqref{eq:DFriedrA}.
\end{proof}

Next we characterise the crucial space $\ker(\mathcal{A}_{\lambda,\alpha}^{(\ell)})^\star$.

\begin{theorem}{Proposition}\label{lem:kerAstar}
Let $m>m^*$, $\ell\geqslant 1$, and $\lambda>0$ large enough (depending on $\alpha$) so as to make the bottom of $\mathcal{A}_{\lambda,\alpha}^{(\ell)}$ strictly positive, as found in Proposition \ref{prop:Apos}. The subspace $\ker(\mathcal{A}_{\lambda,\alpha}^{(\ell)})^\star\subset H^{-1/2}_{W_\lambda,\ell}(\mathbb{R}^3)$ consists of all the functions $\Xi\in H^{-1/2}_{W_\lambda,\ell}(\mathbb{R}^3)$ such that
\begin{equation}\label{eq:AstarXi=0}
2\pi^2\sqrt{\nu p^2+\lambda}\;\widehat{\Xi}(p)+\int_{\mathbb{R}^3}\frac{\widehat{\Xi}(q)}{p^2+q^2+\mu p\cdot q+\lambda}\,\ud q+\alpha\,\widehat{\Xi}(p)\;=\;0
\end{equation}
as a (distributional) identity in $H^{-3/2}(\mathbb{R}^3)$, and then also point-wise almost everywhere in $p\in\mathbb{R}^3$.
\end{theorem}

\begin{proof}
Let $\Xi\in \ker(\mathcal{A}_{\lambda,\alpha}^{(\ell)})^\star$ and let $(\Xi_n)_{n}$ be a sequence in $H^1_\ell(\mathbb{R}^3)$ such that $\Xi_n\rightarrow \Xi$ in the $H^{-1/2}_{\ell}$-norm, equivalently, in the $H^{-1/2}_{W_{\lambda},\ell}$-norm. Then, for arbitrary $\xi\in\mathcal{D}(\mathcal{A}_{\lambda,\alpha}^{(\ell)})= H_{\ell}^{3/2}(\mathbb{R}^3)$,
\begin{align*}
0\;&=\;{\textstyle\frac{1}{2}}\langle\Xi,\mathcal{A}_{\lambda,\alpha}^{(\ell)}\xi\rangle_{H^{-1/2}_{W_\lambda}}\;=\;\langle\Xi,(T_\lambda^{(\ell)}+\alpha\mathbbm{1})\xi\rangle_{H^{-\frac{1}{2}},H^{\frac{1}{2}}}\\ 
&=\;\lim_{n\rightarrow +\infty}\langle\Xi_n,(T_\lambda^{(\ell)}+\alpha\mathbbm{1})\xi\rangle_{H^{-\frac{1}{2}},H^{\frac{1}{2}}}\;=\;\lim_{n\rightarrow +\infty}\langle\Xi_n,(T_\lambda^{(\ell)}+\alpha\mathbbm{1})\xi\rangle_{L^2} \\
&=\;\lim_{n\rightarrow +\infty}\langle (T_\lambda^{(\ell)}+\alpha\mathbbm{1})\Xi_n,\xi\rangle_{L^2}\;=\;\lim_{n\rightarrow +\infty}\langle (T_\lambda^{(\ell)}+\alpha\mathbbm{1})\Xi_n,\xi\rangle_{H^{-\frac{3}{2}},H^{\frac{3}{2}}} \\
&=\;\langle (T_\lambda^{(\ell)}+\alpha\mathbbm{1})\Xi,\xi\rangle_{H^{-\frac{3}{2}},H^{\frac{3}{2}}}\,,
\end{align*}
where in the second step we used \eqref{eq:def_A_ell_geq1}, in the third step we used the strong continuity in the first entry of the dual product, in the fifth step we used the symmetry of the operator $T^{(\ell)}$ on $L^2(\mathbb{R}^3)$ (Proposition \ref{prop:T-W}(ii)), and in the last two steps we used the fact that $T_\lambda^{(\ell)}$, \emph{only for $\ell\geqslant 1$}, maps continuously $H^{-1/2}_{\ell}(\mathbb{R}^3)$ into $H^{-3/2}_{\ell}(\mathbb{R}^3)$ (Proposition \ref{prop:T-W}(v), bound \eqref{eq:T-ell-3/2-1/2_included}). Therefore, by duality, $ (T_\lambda^{(\ell)}+\alpha\mathbbm{1})\Xi=0$ in $H^{-3/2}_{\ell}(\mathbb{R}^3)$, and \eqref{eq:AstarXi=0} follows.
\end{proof}

\begin{remark}{Remark}
In the above proof, the \emph{finiteness} of both the double integrals
\[
\int_{\mathbb{R}^3}\ud p\;\overline{\widehat{\,\Xi}(p)}\,\Big(\int_{\mathbb{R}^3}\ud q\:\frac{\widehat{\xi}(q)}{p^2+q^2+\mu p\cdot q+\lambda}\Big)
\]
and
\[
\int_{\mathbb{R}^3}\ud q\,\Big(\int_{\mathbb{R}^3}\ud p\:\frac{\overline{\,\widehat{\Xi}(p)}}{p^2+q^2+\mu p\cdot q+\lambda}\Big)\,\widehat{\xi}(q)
\]
is a consequence of the properties of $T_\lambda$ when $\ell\geqslant 1$ (the bound \eqref{eq:T-ell-3/2-1/2_included} of Proposition \ref{prop:T-W}(v)), \emph{whereas each of them would diverge for generic $\xi$ or $\Xi$ in $H^{-1/2}(\mathbb{R}^3)$}. The exchange of the order of integration did only take place at the level of suitably regular approximants of $\Xi$ and sufficiently regular $\xi$, which is the (Fubini-based) argument by which $T_\lambda^{(\ell)}$ was proved to be \emph{symmetric} on a regular enough dense of $L^2(\mathbb{R}^3)$. In turn, the approximation was controlled by means of the $H^{-1/2}_{\ell}(\mathbb{R}^3)\to H^{-3/2}_{\ell}(\mathbb{R}^3)$ continuity of $T_\lambda^{(\ell)}$, for which $\ell\geqslant 1$ is crucial.
\end{remark}

We have thus established that $\ker(\mathcal{A}_{\lambda,\alpha}^{(\ell)})^\star$ consists of all those $\Xi$'s that are \emph{distributional} solutions to
\begin{equation}\label{eq:(T+a)Xi=0_H-1/2}
\qquad(T_\lambda^{(\ell)}+\alpha\mathbbm{1})\,\Xi\;=\;0\,,\qquad\qquad \Xi\in H^{-1/2}_{\ell}(\mathbb{R}^3)
\end{equation}
(in this context it is obviously of no relevance to distinguish between $H^{-1/2}_{W_\lambda,\ell}(\mathbb{R}^3)$ and $H^{-1/2}_{\ell}(\mathbb{R}^3)$).

\begin{remark}{Remark}
Combining \eqref{eq:DFriedrA} and \eqref{eq:(T+a)Xi=0_H-1/2} we see that in the decomposition
\[
\mathcal{D}((\mathcal{A}_{\lambda,\alpha}^{(\ell)})^\star)\;=\;\mathcal{D}(\mathcal{A}_{\lambda,\alpha}^{(\ell),F})\dotplus\ker(\mathcal{A}_{\lambda,\alpha}^{(\ell)})^\star
\]
each $\eta\in \mathcal{D}((\mathcal{A}_{\lambda,\alpha}^{(\ell)})^\star)$ decomposes as $\eta=\xi+\Xi$, where $\xi\in\mathcal{D}(\mathcal{A}_{\lambda,\alpha}^{(\ell),F})$ and hence $(\mathcal{A}_{\lambda,\alpha}^{(\ell)})^\star\,\xi=\mathcal{A}_{\lambda,\alpha}^{(\ell),F}\,\xi=2\,W_\lambda^{-1}(T_\lambda^{(\ell)}+\alpha\mathbbm{1})\,\xi$, whereas $\Xi\in\ker(\mathcal{A}_{\lambda,\alpha}^{(\ell)})^\star$ and hence $(T_\lambda^{(\ell)}+\alpha\mathbbm{1})\,\Xi=0$. We may summarise all this by saying that 
\begin{equation}
(\mathcal{A}_{\lambda,\alpha}^{(\ell)})^\star\;=\;2\,W_\lambda^{-1}(T_\lambda^{(\ell)}+\alpha\mathbbm{1})
\end{equation}
as a \emph{distributional} operator.
\end{remark}

\begin{remark}{Remark}\label{rem:remark_history}
At this point it is worth remarking once more the difference between the present rigorous application of the formalism and of the extension theory, and what was done instead in the previous literature. Since we are after $\ker(\mathcal{A}_{\lambda,\alpha}^{(\ell)})^\star$, we are led to the problem \eqref{eq:(T+a)Xi=0_H-1/2} in the unknown $\Xi\in H^{-1/2}_{W_\lambda,\ell}(\mathbb{R}^3)$, for $m>m^*$. In a number of past studies it was erroneously argued that it was \emph{the self-adjoint extensions of $T_\lambda^{(\ell)}+\alpha\mathbbm{1}$ on the Hilbert space $L^2_\ell(\mathbb{R}^3)$} the correct class in one-to-one with the self-adjoint realisations of the corresponding TMS Hamiltonian with parameter $\alpha$. Since  $T_\lambda^{(\ell)}+\alpha\mathbbm{1}$ has a strictly positive bottom as an operator on $L^2_\ell(\mathbb{R}^3)$ when $m>m^*$ (precisely for the same arguments used in the proof of Proposition \ref{prop:Apos}), then its self-adjoint extensions in $L^2_\ell(\mathbb{R}^3)$ are labelled a la Kre{\u\i}n-Vi\v{s}ik-Birman by self-adjoint operators on Hilbert subspaces of $\ker((T_\lambda^{(\ell)})^*+\alpha\mathbbm{1})$, where here $(T_\lambda^{(\ell)})^*$ denotes the adjoint of $T_\lambda^{(\ell)}$ in $L^2_\ell(\mathbb{R}^3)$. 
In a recent series of works by Minlos \cite{Minlos-2011-preprint_May_2010,Minlos-2012-preprint_30sett2011,Minlos-2012-preprint_1nov2012,Minlos-RusMathSurv-2014} the deficiency indices of the operator $T_\lambda^{(\ell)}$ and its adjoint operator $(T_\lambda^{(\ell)})^*$ on $L^2_\ell(\mathbb{R}^3)$ were thoroughly studied, showing that in fact $(T_\lambda^{(\ell)})^*$ acts precisely as $T_\lambda^{(\ell)}$ on a domain in $L^2_\ell(\mathbb{R}^3)$ which, depending on $m$, may or may not be larger than the domain of $T_\lambda^{(\ell)}$ itself. In those works, therefore, the problem that is considered is the \emph{distributional} solutions to
\begin{equation}\label{eq:(T+a)Xi=0_L2}
\qquad(T_\lambda^{(\ell)}+\alpha\mathbbm{1})\,f\;=\;0\,,\qquad\qquad f\in L^2_\ell(\mathbb{R}^3)\,.
\end{equation}
The problems \eqref{eq:(T+a)Xi=0_H-1/2} and \eqref{eq:(T+a)Xi=0_L2} have the very same \emph{distributional} formulation, however the unknown is set on different spaces. It is no surprise, then, that the dimension of the space of solutions to each such problem depends on the mass parameter $m$ in a different manner.
\end{remark}

In the equation \eqref{eq:AstarXi=0}, or also \eqref{eq:(T+a)Xi=0_H-1/2}, at fixed $\ell$ there is at least a natural $(2\ell+1)$-degeneracy, labelled by the `magnetic' quantum number $n\in\{-\ell,\dots,\ell\}$. Indeed, $T^{(\ell)}_{\ell}$ is spherically symmetric and we recall from \eqref{eq:Hs_ell_decomposition} that
\begin{equation*}
H^{-1/2}_{\ell}(\mathbb{R}^3)\;\cong\;L^2(\mathbb{R}^+,(1+r^2)^{-\frac{1}{2}}r^2\,\ud r)\otimes\mathrm{span}\{Y_{\ell,-\ell},\dots,Y_{\ell,\ell}\}\,;
\end{equation*}
thus, if $f\in L^2(\mathbb{R}^+,(1+r^2)^{-\frac{1}{2}}r^2\,\ud r)$ and $\widehat{\Xi}_n(p):=f(|p|)Y_{\ell,n}(\Omega_p)$, then $\Xi_{n^*}$ is a solution to \eqref{eq:AstarXi=0} for some $n^*\in\{-\ell,\dots,\ell\}$ if and only if $\Xi_n$ is a solution of all other $n\in\{-\ell,\dots,\ell\}$. 

\bigskip

We are now in the condition of stating our main conjecture.

\bigskip

\textsc{Main conjecture.} \emph{
%\begin{theorem}{Conjecture}
Let $m>m^*$, $\ell\geqslant 1$, and $\lambda>0$ large enough (depending on $\alpha$) so as to make the bottom of $\mathcal{A}_{\lambda,\alpha}^{(\ell)}$ strictly positive, as found in Proposition \ref{prop:Apos}. There exists a mass threshold $m^{**}>m^*$, expected to be given by the unique root $m^{**}\approx(8.62)^{-1}$ of the integral equation
\begin{equation}\label{eq:m**root}
\pi\sqrt{\frac{\,m(m+2)}{(m+1)^2}}+\int_{-1}^1\!\ud y\,y\int_0^{+\infty}\!\!\ud r\,\frac{r}{\,r^2+1+\frac{2}{m+1}ry}\;=\;0\,,
\end{equation}
with the following property:
\begin{itemize}
 \item if $m\in(m^*,m^{**})$ and $\ell=1$, then equation \eqref{eq:AstarXi=0} has a \emph{unique} non-zero solution in $H^{-1/2}_{\ell}(\mathbb{R}^3)$, modulo the triple angular degeneracy;
 \item if $m\geqslant m^{**}$ and $\ell=1$, as well as if $\ell\geqslant 2$ and $m>m^*$, then equation \eqref{eq:AstarXi=0} has \emph{no} non-zero solutions in $H^{-1/2}_{\ell}(\mathbb{R}^3)$.
\end{itemize}
Equivalently, owing to Proposition \ref{lem:kerAstar},
\begin{equation}
\dim\ker(\mathcal{A}_{\lambda,\alpha}^{(\ell)})^\star\;=\;\begin{cases}
\;\;3 & m\in(m^*,m^{**})\textrm{ and }\ell=1 \\
\;\;0 & m\geqslant m^{**}\textrm{ and }\ell=1\,,\textrm{ or }m>m^*\textrm{ and }\ell\geqslant 2\,.
\end{cases}
\end{equation}
%\end{theorem}
}

% 
% Since the TMS condition \eqref{eq:g_tms_2+1_part2} is formally equivalent to 
% \begin{equation}
% \eta\;=\;2\,W_\lambda^{-1}(T_\lambda+\alpha\mathbbm{1})\,\xi
% \end{equation}

We shall discuss the evidence and the motivation for our conjecture later in Section~\ref{sec:mass_thr_and_evidence_of_conj} In the remaining part of this Section let us examine how the self-adjoint extension scheme for $\mathcal{A}_{\lambda,\alpha}^{(\ell)}$ is completed, based upon our main conjecture.

\bigskip

\underline{\textbf{Case $m\geqslant m^{**}$ and $\ell=1$, or $m>m^*$ and $\ell\geqslant 2$.}} 

Each $\mathcal{A}_{\lambda,\alpha}^{(\ell)}$, for any $\ell\geqslant 1$, is (essentially) self-adjoint on $H^{-1/2}_{W_\lambda,\ell}(\mathbb{R}^3)$, its closure being the Friedrichs extension $\mathcal{A}_{\lambda,\alpha}^{(\ell),F}$ characterised in Proposition \ref{prop:Friedrichs_ext_of_A}.

\bigskip

\underline{\textbf{Case $m\in(m^*,m^{**})$ and $\ell=1$.}} 

$\mathcal{A}_{\lambda,\alpha}^{(1)}$ has deficiency index equal to $3$ and its self-adjoint extensions on $H^{-1/2}_{W_\lambda,\ell=1}(\mathbb{R}^3)$ are one-to-one with self-adjoint operators acting on Hilbert subspaces of $\ker(\mathcal{A}_{\lambda,\alpha}^{(1)})^\star$.
Since the Friedrichs extension $\mathcal{A}_{\lambda,\alpha}^{(1),F}$ and $\mathcal{A}_{\lambda,\alpha}^{(1)}$ have the same strictly positive bottom, $(\mathcal{A}_{\lambda,\alpha}^{(1),F})^{-1}$ exists and is bounded and everywhere defined on $H^{-1/2}_{W_\lambda,1}(\mathbb{R}^3)$.
We denote by $\Xi_{1,n}$, $n\in\{-1,0,1\}$, the unique (up to multiples) solution to equation \eqref{eq:AstarXi=0} with angular symmetry $(\ell=1,n)$. Clearly, the $\Xi_{1,n}$'s span $\ker(\mathcal{A}_{\lambda,\alpha}^{(1)})^\star$, and the span of each $\Xi_{1,n}$ is a reducing subspace for both both $(\mathcal{A}_{\lambda,\alpha}^{(1)})^\star$ and $(\mathcal{A}_{\lambda,\alpha}^{(1),F})^{-1}$. Also, on $\mathrm{span}\,\{\Xi_{1,n}\}$ the only self-adjoint operators are multiplications by a real number, say, $\beta_n\in\mathbb{R}$.

\begin{theorem}{Proposition}
Let $m\in(m^*,m^{**})$ and $\lambda>0$ large enough (depending on $\alpha$) so as to make the bottom of $\mathcal{A}_{\lambda,\alpha}^{(1)}$ strictly positive, as found in Proposition \ref{prop:Apos}.
The family of self-adjoint extensions of $\mathcal{A}_{\lambda,\alpha}^{(1)}$ is a three-real-parameter family of operators
\begin{equation}\label{eq:A-extensions_op_family}
\big\{ \,\mathcal{A}_{\lambda,\alpha}^{(1),\beta}\,|\,\beta\equiv(\beta_{-1},\beta_0,\beta_1)\in(-\infty,+\infty]^3\,\big\}
\end{equation}
with
\begin{equation}\label{eq:A-extensions_op}
\begin{split}
\mathcal{D}(\mathcal{A}_{\lambda,\alpha}^{(1),\beta})\;&=\;\left\{\!\!
\begin{array}{c}
\xi\;=\;\xi_1+(\mathcal{A}_{\lambda,\alpha}^{(1),F})^{-1}\Big(\!\displaystyle\sum_{n=-1}^1 \beta_n q_n\Xi_{1,n}\Big)+\sum_{n=-1}^1 q_n\Xi_{1,n} \\
\textrm{where }\xi_1\in\mathcal{D}(\overline{\mathcal{A}_{\lambda,\alpha}^{(1)}})\textrm{ and }(q_{-1},q_0, q_1)\in\mathbb{C}^{3}
\end{array}\!\!\right\} \\
\mathcal{A}_{\lambda,\alpha}^{(1),\beta}\,\xi\;&=\;\overline{\mathcal{A}_{\lambda,\alpha}^{(1)}}\:\xi_1+\displaystyle\sum_{n=-1}^1 \beta_n q_n\Xi_{1,n}\,.
\end{split}
\end{equation}
Each `charge' $\xi\in\mathcal{D}(\mathcal{A}_{\lambda,\alpha}^{(1),\beta})$ consists of a `regular' component
\begin{equation}\label{eq:xi_reg}
\begin{split}
\xi_{\mathrm{reg}}\;&:=\;\xi_1+(\mathcal{A}_{\lambda,\alpha}^{(1),F})^{-1}\big(\!\sum_{n=-1}^1 \beta_n q_n\Xi_{1,n}\big)\;\in\;\mathcal{D}(\mathcal{A}_{\lambda,\alpha}^{(1),F}) \\
\mathcal{A}_{\lambda,\alpha}^{(1),\beta}\,\xi\;&=\;\mathcal{A}_{\lambda,\alpha}^{(1),F}\,\xi_{\mathrm{reg}}\,,
\end{split}
\end{equation}
plus a `singular' component
\[
\xi_{\mathrm{sing}}\;:=\;\sum_{n=-1}^1 q_n\Xi_{1,n}\;\in\;\ker(\mathcal{D}(\mathcal{A}_{\lambda,\alpha}^{(1)}))^\star\,.
\]
In particular, $\xi_{\mathrm{reg}}$ has the following properties:
\begin{equation}\label{eq:properties_of_xi_in_DAb1}
\begin{split}
&\xi_{\mathrm{reg}}\in H^{1/2}_{\ell=1}(\mathbb{R}^3)\,,\qquad(T_\lambda^{(1)}+\alpha\mathbbm{1})\,\xi_{\mathrm{reg}}\in  H^{1/2}_{\ell=1}(\mathbb{R}^3)\,, \\
2\,\langle\Xi_{1,n},&(T_\lambda^{(1)}+\alpha\mathbbm{1})\,\xi_{\mathrm{reg}}\rangle_{H^{-\frac{1}{2}},H^{\frac{1}{2}}}\;=\;\beta_n q_n\|\Xi_{1,n}\|^2_{H^{-1/2}_{W_\lambda}}\qquad n\in\{-1,0,1\}\,.
\end{split}
\end{equation}
The quadratic form of each such extension is given by
\begin{equation}\label{eq:quadratic_forms_Abeta_v2}
\begin{split}
\mathcal{D}[\mathcal{A}_{\lambda,\alpha}^{(1),\beta}]\;&=\;H^{1/2}_{\ell=1}(\mathbb{R}^3)\:\dotplus\:\mathrm{span}\{\Xi_{1,-1},\Xi_{1,0},\Xi_{1,1}\}\\
\mathcal{A}_{\lambda,\alpha}^{(1),\beta}\Big[\xi+\sum_{n=-1}^1 q_n\Xi_{1,n}\Big]\;&= \\
&\hspace{-2cm}=\;2\,\langle\xi,(T_\lambda^{(1)}+\alpha\mathbbm{1})\xi\rangle_{H^{\frac{1}{2}},H^{-\frac{1}{2}}}+\sum_{n=-1}^1 \beta_n|q_n|^2\|\Xi_{1,n}\|^2_{H^{-1/2}_{W_\lambda}}\,.
\end{split}
\end{equation}
\end{theorem}

\begin{proof}
The three-parameter nature of the family \eqref{eq:A-extensions_op_family} for $m\in(m^*,m^{**})$ is an immediate consequence of the dimensionality of $\ker(\mathcal{A}_{\lambda,\alpha}^{(1)})^\star$, and \eqref{eq:A-extensions_op} is then the explicit Kre{\u\i}n-Vi\v{s}ik-Birman extension formula applied to this case (we refer to \cite[Theorem 3.4]{M-KVB2015}, for concreteness). By construction the vector in \eqref{eq:xi_reg} belongs to the domain of the Friedrichs extension $\mathcal{A}_{\lambda,\alpha}^{(1),F}$, thus the first two properties of \eqref{eq:properties_of_xi_in_DAb1} follow at once from \eqref{eq:DFriedrA}. Moreover,
\[
\begin{split}
2\,\langle&\Xi_{1,n},(T_\lambda^{(1)}+\alpha\mathbbm{1})\,\xi_{\mathrm{reg}}\rangle_{H^{-\frac{1}{2}},H^{\frac{1}{2}}}\;=\;\langle \Xi_{1,n},\mathcal{A}_{\lambda,\alpha}^{(1),F}\,\xi_{\mathrm{reg}}\rangle_{H^{-1/2}_{W_\lambda}} \\
&=\;\langle \Xi_{1,n},\mathcal{A}_{\lambda,\alpha}^{(1),\beta}\,\xi\rangle_{H^{-1/2}_{W_\lambda}}\;=\;\langle \Xi_{1,n}\overline{\mathcal{A}_{\lambda,\alpha}^{(1)}}\:\xi_1\rangle_{H^{-1/2}_{W_\lambda}}+\displaystyle\sum_{a=-1}^1 \beta_a q_a\langle\Xi_{1,n},\Xi_{1,a}\rangle_{H^{-1/2}_{W_\lambda}} \\
&=\;\beta_n q_n\|\Xi_{1,n}\|^2_{H^{-1/2}_{W_\lambda}}
\end{split}
\]
where we used \eqref{eq:DFriedrA} in the first step, \eqref{eq:xi_reg} in the second, \eqref{eq:A-extensions_op} in the third, and the $W_\lambda$-orthogonality of the $\Xi_{1,a}$'s, together with $(\mathcal{A}_{\lambda,\alpha}^{(1)})^\star\Xi_{1,a}=0$, in the fourth. This yields the last property of \eqref{eq:properties_of_xi_in_DAb1}. Next to the classification \eqref{eq:A-extensions_op}, the Kre{\u\i}n-Vi\v{s}ik-Birman theory provides also a convenient parametrisation of the self-adjoint extensions of $\mathcal{A}_{\lambda,\alpha}^{(1)}$ in terms of their quadratic forms. Referring, for concreteness, to \cite[Theorem 3.4]{M-KVB2015}, we find
\begin{equation*}%\label{eq:quadratic_forms_Abeta}
\begin{split}
\mathcal{D}[\mathcal{A}_{\lambda,\alpha}^{(1),\beta}]\;&=\;\mathcal{D}[\mathcal{A}_{\lambda,\alpha}^{(1),F}]\:\dotplus\:\ker(\mathcal{D}(\mathcal{A}_{\lambda,\alpha}^{(1)}))^\star \\
\mathcal{A}_{\lambda,\alpha}^{(1),\beta}\Big[\xi+\sum_{n=-1}^1 q_n\Xi_{1,n}\Big]\;&=\;\mathcal{A}_{\lambda,\alpha}^{(1)}[\xi]+\sum_{n=-1}^1 \beta_n|q_n|^2\|\Xi_{1,n}\|^2_{H^{-1/2}_{W_\lambda}}\,.
\end{split}
\end{equation*}
Owing to \eqref{eq:DformA}, this reads precisely as \eqref{eq:quadratic_forms_Abeta_v2}.
\end{proof}

Having conjectured that $\ker(\mathcal{D}(\mathcal{A}_{\lambda,\alpha}^{(1)}))^\star$ is finite-dimensional, necessarily each $\mathcal{A}_{\lambda,\alpha}^{(1),\beta}$ is semi-bounded below, \cite[Corollary 5.7]{M-KVB2015}.

The Friedrichs extension corresponds, in the parametrisation \eqref{eq:A-extensions_op}, to the case ``$\beta=\infty$'' \cite[Proposition 3.7]{M-KVB2015}. It is characterised in Proposition \ref{prop:Friedrichs_ext_of_A} above.

 %we re-write \eqref{eq:quadratic_forms_Abeta} as

\section{Auxiliary operator $\mathcal{A}_{\lambda,\alpha}$ in the sector $\ell=0$}\label{sec:A_ell_0}

So far, on the sector of symmetry $\ell=0$ the `Kre{\u\i}n-Vi\v{s}ik-Birman extension parameter' $\mathcal{A}_{\lambda,\alpha}$ was generically qualified  to be a self-adjoint operator $\mathcal{A}_{\lambda,\alpha}^{(\ell=0)}$ on $H^{-1/2}_{W_\lambda,\ell=0}(\mathbb{R}^3)$. %, in order to refrain ourselves from using the same structure $W_\lambda^{-1}(T_\lambda+\alpha\mathbbm{1})$ 

The peculiarity of the sector $\ell=0$ is that the bound \eqref{eq:T-ell-3/2-1/2_included}, namely the $H^{3/2}_{\ell}(\mathbb{R}^3)\to H^{1/2}_{\ell}(\mathbb{R}^3)$ continuity of $T_\lambda^{(\ell)}$, is lost, and in fact on the opposite \eqref{eq:a_dense_not_in_it} holds, which makes the composition $W_\lambda^{-1} T_\lambda$ in general ill-defined on spherically symmetric functions -- for sure ill-defined on a very much natural dense of $H^{-1/2}_{W_\lambda,\ell=0}(\mathbb{R}^3)$ !

On the other hand, in our discussion on the sectors $\ell\geqslant 2$ the essential feature of the extension parameter $\mathcal{A}_{\lambda,\alpha}^{(\ell)}$ was to have a Friedrichs extension whose quadratic form \eqref{eq:DformA} takes the expression $\xi\mapsto 2\,\langle\,\xi,(T_\lambda^{(\ell)}+\alpha\mathbbm{1})\xi\rangle_{H^{\frac{1}{2}},H^{-\frac{1}{2}}}$.

This motivates the following construction of $\mathcal{A}_{\lambda,\alpha}^{(\ell=0)}$ .

%Two ``extreme'' choices for $\mathcal{A}_{\lambda,\alpha}^{(\ell=0)}$ would be $\mathcal{A}_{\lambda,\alpha}^{(\ell=0)}=\mathbb{O}$ or $\mathcal{A}_{\lambda,\alpha}^{(\ell=0)}=\alpha\mathbbm{1}$: we shall comment in Section \ref{sec:reconstr_Halpha} what kind of two-body short-scale asymptotics each of them results in for the functions of the corresponding self-adjoint Hamiltonian $H_{\alpha}$.

%A different, non-trivial choice for $\mathcal{A}_{\lambda,\alpha}^{(\ell=0)}$, motivated by making the present discussion directly comparable with the existing literature, is the one emerging from the following simple fact.

\begin{theorem}{Proposition}\label{prop:ext_param_l=0}
Let $m>m^*$ and $\lambda>0$ large enough. On the Hilbert space $H^{-1/2}_{W_\lambda,\ell=0}(\mathbb{R}^3)$ the expression
\begin{equation}\label{eq:qform}
\begin{split}
q[\xi_1,\xi_2]\;&:=\;2\,\langle\xi_1,(T^{(0)}_\lambda+\alpha\mathbbm{1})\xi_2)\rangle_{H^{\frac{1}{2}},H^{-\frac{1}{2}}} \\
\xi_1,\xi_2\in\mathcal{D}(q)\;&:=\;H^{1/2}_{\ell=0}(\mathbb{R}^3)
\end{split}
\end{equation}
defines a densely defined, positive definite, and closed quadratic form, with strictly positive bottom. Let $\mathcal{A}_{\lambda,\alpha}^{(\ell=0)}$ be the unique self-adjoint operator on $H^{-1/2}_{W_\lambda,\ell=0}(\mathbb{R}^3)$ that realises such a form, that is, 
\begin{equation}\label{eq:Al0}
\begin{split}
\mathcal{D}[\mathcal{A}_{\lambda,\alpha}^{(\ell=0)}]\;&=\;H^{1/2}_{\ell=0}(\mathbb{R}^3)  \\
\mathcal{A}_{\lambda,\alpha}^{(\ell=0)}[\xi_1,\xi_2]\;&=\;2\,\langle\xi_1,(T^{(0)}_\lambda+\alpha\mathbbm{1})\xi_2)\rangle_{H^{\frac{1}{2}},H^{-\frac{1}{2}}}\,.
\end{split}
\end{equation}
Then
\begin{equation}\label{eq:DA_op_l0}
\begin{split}
\mathcal{D}(\mathcal{A}_{\lambda,\alpha}^{(\ell=0)})\;&=\;\{\xi\in H^{1/2}_{\ell=0}(\mathbb{R}^3)\,|\,(T_\lambda^{(0)}+\alpha\mathbbm{1})\xi\in  H^{1/2}_{\ell=0}(\mathbb{R}^3)\} \\
\mathcal{A}_{\lambda,\alpha}^{(\ell=0)}\,\xi\;&=\;2\,W_\lambda^{-1}(T_\lambda^{(0)}+\alpha\mathbbm{1})\,\xi\,.
\end{split}
\end{equation}
\end{theorem}

\begin{proof}
The density of $\mathcal{D}(q)$ in  $H^{-1/2}_{W_\lambda,\ell=0}(\mathbb{R}^3)$ is obvious, and the boundedness from below of $q$ with strictly positive bottom follows by precisely the same argument used in the proof of Proposition \ref{prop:Apos}: in fact, the bound \eqref{eq:positivity_from_CDFMT2012} was proved in \cite[Proposition 3.1]{CDFMT-2012} also for $\ell=0$. Moreover, analogously to what argued in the proof of Proposition \ref{prop:Friedrichs_ext_of_A}, $q[\xi,\xi]$ defines a norm in $H^{-1/2}_{W_\lambda,\ell=0}(\mathbb{R}^3)$ which is equivalent to the $H^{1/2}$-norm, as a consequence of the bounds \cite[Eq.(3.52)]{CDFMT-2012}. Thus, $\mathcal{D}(q)$ is complete in the norm $q[\xi,\xi]$, which means that $q$ is closed. It is therefore standard that there is a unique self-adjoint operator $\mathcal{A}_{\lambda,\alpha}^{(\ell=0)}$ whose form domain is precisely  $H^{1/2}_{\ell=0}(\mathbb{R}^3)$ and $\mathcal{A}_{\lambda,\alpha}^{(\ell=0)}[\xi_1,\xi_2]=q[\xi_1,\xi_2]$. The operator domain and the operator action \eqref{eq:DA_op_l0} are deduced from the form domain and the form action \eqref{eq:Al0} exactly the same way we derived \eqref{eq:DFriedrA} from \eqref{eq:DformA} in the proof of Proposition \ref{eq:DFriedrA}: one can indeed repeat verbatim the same argument, for there was nothing constraining that part of the proof of Proposition \ref{eq:DFriedrA} to non-zero $\ell$'s.
\end{proof}

\section{Reconstruction of the self-adjoint TMS Hamiltonians}\label{sec:reconstr_Halpha}

We analyse in this Section the explicit structure of the self-adjoint extensions of $\mathring{H}$ of the class $\mathring{H}_{\alpha}$ given by \eqref{eq:D_Halpha_2+1}, that is, those extensions of TMS type. We show that, following our main conjecture stated in Section \ref{sec:ext_scheme_for_A}, \emph{we retrieve precisely the family of TMS self-adjoint Hamiltonians derived in the physical literature, consistently also with previous partial attempts in the mathematical literature developed in the language of the quadratic forms}.

In this Section too we assume $m>m^*$ and, for arbitrary $\alpha$, we assume $\lambda$ to be large enough, in the sense of Propositions \ref{prop:Apos} and \ref{lem:kerAstar}.

In view of the discussion of Sections \ref{sec:ext_scheme_for_H-TMS} and \ref{sec:ext_scheme_for_A} above, we are led to consider the self-adjoint extensions $\mathring{H}^{[\beta]}_{\alpha}$, $\beta\in\mathbb{R}^3$, whose extension parameter on the space of charges $H^{-1/2}_{W_\lambda}(\mathbb{R}^3)=\bigoplus_{\ell\geqslant 0}H^{-1/2}_{W_\lambda,\ell}(\mathbb{R}^3)$ is the self-adjoint operator
\begin{equation}\label{eq:ext_param_A_lab}
\mathcal{A}_{\lambda,\alpha}^{[\beta]}\;=\;\mathcal{A}_{\lambda,\alpha}^{(\ell=0)}\;\oplus\;\mathcal{A}_{\lambda,\alpha}^{(1),\beta}\oplus\;\big({\textstyle\bigoplus_{\ell=2}^\infty}\,\overline{\,\mathcal{A}_{\lambda,\alpha}^{(\ell)}}\,\big)\,.
\end{equation}
It is understood, in view of our main conjecture, that on the $\ell=1$ sector there is a multiplicity of choices $\beta\in(-\infty,+\infty]^3$ only in the regime $m\in(m^*,m^{**})$, whereas for $m\geqslant m^{**}$ the only choice ($\beta=\infty$) is operator closure of $\mathcal{A}_{\lambda,\alpha}^{(1)}$, namely its unique self-adjoint realisation. Moreover, due to our conjecture, on the $\ell\geqslant 2$ sectors the operator closure of $\mathcal{A}_{\lambda,\alpha}^{(\ell)}$ is its unique self-adjoint extension. On the $\ell=0$ sector $\mathcal{A}_{\lambda,\alpha}^{(\ell=0)}$ is by construction already self-adjoint. % we take a generic self-adjoint operator, such as those considered in Section \ref{sec:A_ell_0}.

%When $\ell\geqslant 2$ 

%In other words: $\mathcal{A}_{\lambda,\alpha}^{[\beta]}$ consists of the unique self-adjoint extension of $\mathcal{A}_{\lambda,\alpha}^{(\ell)}$, namely, its operator closure, on the $\ell$-th symmetry sector $H^{-1/2}_{W_\lambda,\ell}(\mathbb{R}^3)$ of $H^{-1/2}_{W_\lambda}(\mathbb{R}^3)$ for any $\ell\geqslant 2$; on the $\ell=1$ sector $\mathcal{A}_{\lambda,\alpha}^{[\beta]}$ is instead the self-adjoint extension $\mathcal{A}_{\lambda,\alpha}^{(1),\beta}$ of $\mathcal{A}_{\lambda,\alpha}^{(1)}$ at the given parameter $\beta$, where it is understood, in view of our main conjecture, that there is a multiplicity of choices $\beta\in\mathbb{R}^3$ only in the regime $m\in(m^*,m^{**})$, whereas for $m>m^{**}$ the only choice ($\beta=\infty$) is operator closure of $\mathcal{A}_{\lambda,\alpha}^{(1)}$, namely its unique self-adjoint realisation; last, on the $\ell=0$ sector, 

We observe that $\{\mathcal{A}_{\lambda,\alpha}^{[\beta]}\,|\,\beta\in(-\infty,+\infty]^3\}$ is the family of self-adjoint extensions of the operator $\mathcal{A}_{\lambda,\alpha}$ defined in \eqref{eq:def_A_all_sectors} on the Hilbert space $H^{-1/2}_{W_\lambda}(\mathbb{R}^3)$, with the specification of $\mathcal{A}_{\lambda,\alpha}^{(\ell=0)}$ made in Section \ref{sec:A_ell_0}.

Corresponding to each $\mathcal{A}_{\lambda,\alpha}^{[\beta]}$, the Kre{\u\i}n-Vi\v{s}ik-Birman theory associates a Hamiltonian $H_{\alpha}^{[\beta]}$ on the physical Hilbert space $\cH=L^2_\mathrm{f}(\mathbb{R}^3\times\mathbb{R}^3,\ud x_1\ud x_2)$, which is a self-adjoint extension of the initial operator $\mathring{H}$ and is given by the extension formula (see \eqref{eq:D_Halpha_2+1} and \eqref{eq:ext_formula_etaAxi} above)
\begin{equation}\label{eq:def_Halpha_beta}
\begin{split}
\mathcal{D}(H_{\alpha}^{[\beta]})\;&:=\;
\left\{\,g=f+(\mathring{H}_F+\lambda\mathbbm{1})^{-1}\big(u_{\mathcal{A}_{\lambda,\alpha}^{[\beta]}\xi}\big)+u_\xi\left|
\begin{array}{c}
f\in\mathcal{D}(\mathring{H}) \\
\xi\in\mathcal{D}(\mathcal{A}_{\lambda,\alpha}^{[\beta]})
\end{array}\!\!\right.\right\}\,, \\
% &=\left\{
% \begin{array}{c}
% g=f+(\mathring{H}_F+\lambda\mathbbm{1})^{-1}\big((U_\lambda^{-1}\mathcal{A}_{\lambda,\alpha}^{[\beta]} U_\lambda )u_\xi\big)+u_\xi \\
% f\in\mathcal{D}(\mathring{H})\quad\textrm{and}\quad u_\xi\in U_\lambda^{-1}\mathcal{D}(\mathcal{A}_{\lambda,\alpha}^{[\beta]})
% \end{array}\right\}\,, \\
&\:=\;\,\big\{g\in\mathcal{D}(\mathring{H}^*)\,|\,\eta=\mathcal{A}_{\lambda,\alpha}^{[\beta]}\,\xi\big\} \\
H_{\alpha}^{[\beta]}\;&:=\;\mathring{H}^*\upharpoonright\mathcal{D}(H_{\alpha}^{[\beta]})\,.
\end{split}
\end{equation}
%Recall that $U_\lambda:\ker (\mathring{H}^*+\lambda\mathbbm{1})\xrightarrow[]{\cong}H^{-1/2}_{W_\lambda}(\mathbb{R}^3)$ is the unitary isomorphism \eqref{eq:isomorphism_Ulambda}, thus $U_\lambda^{-1}\mathcal{A}_{\lambda,\alpha}^{[\beta]} U_\lambda$ is the self-adjoint operator on the Hilbert space $\ker (\mathring{H}^*+\lambda\mathbbm{1})$ which is unitarily equivalent to $\mathcal{A}_{\lambda,\alpha}^{[\beta]}$.

The Kre{\u\i}n-Vi\v{s}ik-Birman theory provides another extension formula (see, e.g., \cite[Theorem 3.4]{M-KVB2015})
%, the same we used in \eqref{eq:quadratic_forms_Abeta} for the quadratic form of $\mathcal{A}_{\lambda,\alpha}^{(1),\beta}$, which 
that gives the quadratic form of $H_{\alpha}^{[\beta]}$. We find:
\begin{equation}\label{eq:def_Halpha_beta_qform}
\begin{split}
\mathcal{D}[H_{\alpha}^{[\beta]}]\;&:=\;\mathcal{D}[\mathring{H}_F]\dotplus \mathcal{D}[U_\lambda^{-1}\mathcal{A}_{\lambda,\alpha}^{[\beta]} U_\lambda] \\
&=\;H^1_f(\mathbb{R}^3\times\mathbb{R}^3)\dotplus U_\lambda^{-1}\mathcal{D}[\mathcal{A}_{\lambda,\alpha}^{[\beta]}] \\
(H_{\alpha}^{[\beta]}+\lambda\mathbbm{1})[F+u_\xi]\;&:=\;(\mathring{H}_F+\lambda\mathbbm{1})[F]+\mathcal{A}_{\lambda,\alpha}^{[\beta]}[
\xi]\,.
\end{split}
\end{equation}
Recall that $U_\lambda:\ker (\mathring{H}^*+\lambda\mathbbm{1})\xrightarrow[]{\cong}H^{-1/2}_{W_\lambda}(\mathbb{R}^3)$ is the unitary isomorphism \eqref{eq:isomorphism_Ulambda}, thus $U_\lambda^{-1}\mathcal{A}_{\lambda,\alpha}^{[\beta]} U_\lambda$ is the self-adjoint operator on the Hilbert space $\ker (\mathring{H}^*+\lambda\mathbbm{1})$ which is unitarily equivalent to $\mathcal{A}_{\lambda,\alpha}^{[\beta]}$, and $(U_\lambda^{-1}\mathcal{A}_{\lambda,\alpha}^{[\beta]} U_\lambda)[u_\xi]=\mathcal{A}_{\lambda,\alpha}^{[\beta]}[\xi]$.

Whereas
\begin{equation}\label{eq:saext}
\mathring{H}\;\subset\; H_{\alpha}^{[\beta]}\;=\;H_{\alpha}^{[\beta]*}\;\subset\;\mathring{H}^*
\end{equation}
is thus guaranteed by the Kre{\u\i}n-Vi\v{s}ik-Birman extension theory, we want now to make the structure of $H_{\alpha}^{[\beta]}$ explicit both in the operator and in the quadratic form sense, and we want also to discuss the TMS nature of each such extension.

The results are cast in the following two Theorems. 

\begin{theorem}{Theorem}\label{thm:DHab}
With the notation and under the assumptions stated above in this Section, the operator domain of the self-adjoint extension $H_{\alpha}^{[\beta]}$ is the space
\begin{equation}\label{eq:DHab}
\begin{split}
\mathcal{D}(H_{\alpha}^{[\beta]})\;&=\;\left\{\,g=F+u_\xi\left|\!
\begin{array}{c}
F \in H^2_{\mathrm{f}}(\mathbb{R}^3\times\mathbb{R}^3),\;\xi=\xi_{\mathrm{reg}}+\!\displaystyle\sum_{n=-1}^1 q_n\Xi_{1,n}\,, \\
q\equiv(q_{-1},q_0,q_1)\in\mathbb{C}^3,\\
\xi_{\mathrm{reg}}\in H^{\frac{1}{2}}(\mathbb{R}^3)\,,\;(T_\lambda+\alpha)\xi_{\mathrm{reg}}\in H^{\frac{1}{2}}(\mathbb{R}^3) \\
\textrm{plus boundary conditions $(\textsc{bc}1)$ and $(\textsc{bc}2)$}
\end{array}\!\!\!\right.\right\} \\
(\textsc{bc}1) & \quad \int_{\mathbb{R}^3}\widehat{F}(p_1,p_2)\,\ud p_2\;=\;(T_\lambda+\alpha)\,\xi_{\mathrm{reg}}(p_1)\,, \\
(\textsc{bc}2) & \quad
\displaystyle\int_{\mathbb{R}^3}\overline{\,\widehat{\Xi}_{1,n}(p)}\,\big(\widehat{(T_\lambda\xi_{\mathrm{reg}})}(p)+\alpha\,\widehat{\xi_{\mathrm{reg}}}(p)\big)\,\ud p\;=\;\beta_n|q_n|^2\,, \quad n\in\{-1,0,1\}\,,
% \begin{array}{l}
% \displaystyle\int_{\mathbb{R}^3}\overline{\,\widehat{\Xi}_{1,n}(p)}\,\big(\widehat{(T_\lambda\xi_{\mathrm{reg}})}(p)+\alpha\,\widehat{\xi_{\mathrm{reg}}}(p)\big)\,\ud p\;=\;\beta_n|q_n|^2  \\ %\|\Xi_{1,n}\|^2_{H^{-1/2}_{W_\lambda}}
% \textrm{for }n\in\{-1,0,1\},\;\;\textrm{valid only if $\xi_{\mathrm{reg}}$ has a component in $H^{1/2}_{\ell=1}$}
% \end{array}
\end{split}
\end{equation}
and the action of $H_{\alpha}^{[\beta]}$ is given by
\begin{equation}\label{eq:action_DHab}
\begin{split}
(H_{\alpha}^{[\beta]}+\lambda\mathbbm{1})\,g\;&=\;(H_{\mathrm{free}}+\lambda\mathbbm{1})\,F \\
&=\;\big(-\Delta_{x_1}-\Delta_{x_2}-{\textstyle\frac{2}{m+1}}\nabla_{x_1}\cdot\nabla_{x_2}+\lambda\big)\,F\,.
\end{split}
\end{equation}
The three functions $\Xi_{1,-1},\Xi_{1,0},\Xi_{1,1}\in H^{-1/2}_{\ell=1}(\mathbb{R}^3)$ in \eqref{eq:DHab} are those conjectured in our main conjecture and they are taken here to be normalised as $\|\Xi_{1,n}\|_{H^{-1/2}_{W_\lambda}}^2=2$. 
% The boundary condition $(\textsc{bc}1)$ is a Ter-Martirosyan--Skornyakov condition:
% \begin{equation}\label{eq:TMS_cond_asymptotics_2+1_for_Hab}
% \begin{split}
% \int_{\substack{ \\ \,p_2\in\mathbb{R}^3 \\ \! |p_2|<R}}{\:\widehat{g}(p_1,p_2) \,\ud p_2}\;&=\;\widehat{\xi}(p_1)(4\pi R+\alpha)+o(1) \quad\quad (\,\xi\in H^{3/2}_\ell(\mathbb{R}^3)\,)\,.
% %& \qquad\qquad\textrm{as}\qquad R\to +\infty\,.
% \end{split}
% \end{equation}
% *****
\end{theorem}

\begin{theorem}{Theorem}\label{thm:DHab_form}
With the notation and under the assumptions stated above in this Section, the quadratic form of the self-adjoint extension $H_{\alpha}^{[\beta]}$ is given by
\begin{equation}\label{eq:DHab_form}
\begin{split}
\mathcal{D}[H_{\alpha}^{[\beta]}]\;&=\;\left\{g=F+u_\xi\left|\!\!
\begin{array}{c}
F \in H^1_{\mathrm{f}}(\mathbb{R}^3\times\mathbb{R}^3),\;\xi=\xi_{\mathrm{reg}}+\!\displaystyle\sum_{n=-1}^1 q_n\Xi_{1,n}, \\
\xi_{\mathrm{reg}}\in H^{\frac{1}{2}}(\mathbb{R}^3)\,,\;q\equiv(q_{-1},q_0,q_1)\in\mathbb{C}^3
\end{array}\!\!\!\!\right.\right\} \\
H_{\alpha}^{[\beta]}[F+u_\xi]\;&=\;\lambda\|F\|_\cH^2-\lambda\|F+u_\xi\|_\cH^2+H_\mathrm{free}[F] \\
& \!\!\!\!\!\!\!\!+2\,\Big(\int_{\mathbb{R}^3}\overline{\,\widehat{\xi}_{\mathrm{reg}}(p)}\,\big(\widehat{(T_\lambda\xi_{\mathrm{reg}})}(p)+\alpha\,\widehat{\xi_{\mathrm{reg}}}(p)\big)\,\ud p
%\langle\xi,(T_\lambda^{(1)}+\alpha\mathbbm{1})\xi\rangle_{H^{\frac{1}{2}},H^{-\frac{1}{2}}}
+\sum_{n=-1}^1 \beta_n|q_n|^2\Big)\,.
\end{split}
\end{equation}
The three functions $\Xi_{1,-1},\Xi_{1,0},\Xi_{1,1}\in H^{-1/2}_{\ell=1}(\mathbb{R}^3)$ in \eqref{eq:DHab} are those conjectured in our main conjecture and they are taken here to be normalised as $\|\Xi_{1,n}\|_{H^{-1/2}_{W_\lambda}}^2=2$. 
\end{theorem}

As remarked above, there is an actual multiplicity of extensions $H_{\alpha}^{[\beta]}$ only in the regime $m\in(m^*,m^{**})$, whereas for $m\geqslant m^{**}$ the only self-adjoint Hamiltonian is $H_{\alpha}^{[\beta=\infty]}$. In the above formulas, $H_{\alpha}^{[\beta=\infty]}$ is the Hamiltonian 
%whose operator (or form) domain has only regular charges and only the boundary condition (\textsc{bc}1) is present.
whose form domain has regular $H^1$-functions and $H^{\frac{1}{2}}$-charges only, and whose operator domain does not have singular charges in the $\ell=1$ sector, and for which therefore the boundary condition (\textsc{bc2}) is completely absent.

Before discussing the proofs, let us emphasize that Theorems \ref{thm:DHab} and \ref{thm:DHab_form} above are explicitly formulated so as to present the structure of $H_{\alpha}^{[\beta]}$ (as an operator and as a quadratic form) in such a way to allow a most direct comparison between our analysis and the existing literature.

The first rigorous construction of operators of the class $H_{\alpha}^{[\beta]}$ for the (2+1)-fermionic model, which also proved the fact that $H_{\alpha}^{[\beta]}$ is a self-adjoint extension of the ``away-from-hyperplanes'' free Hamiltonian $\mathring{H}$ in the sense of \eqref{eq:saext}, and that it supports a Ter-Martirosyan--Skornyakov condition in its domain, is due to a work of Finco and Teta \cite{Finco-Teta-2012} and a subsequent work of one of us in collaboration with Correggi, Dell'Antonio, Finco, and Teta \cite{CDFMT-2012}, both from 2012,  \emph{relative only to the case $\beta=\infty$} of the present notation. With reference to \eqref{eq:DHab} and \eqref{eq:DHab_form} above, this was precisely the operator $H_\alpha\equiv H_{\alpha}^{[\beta=\infty]}$.
%whose form domain has regular $H^1$-functions and $H^{\frac{1}{2}}$-charges only, and whose operator domain does not have singular charges in the $\ell=1$ sector, and for which therefore the boundary condition (\textsc{bc2}) is completely absent. 

In both works \cite{Finco-Teta-2012,CDFMT-2012} the approach was through quadratic form theory, thus first it was proved that the form $H_\alpha[\,\cdot\,]$ is closed and semi-bounded on $\cH$ and then the corresponding self-adjoint operator was derived and \eqref{eq:saext} proved. For a convenient comparison, our ubiquitous expression $T_\lambda+\alpha\mathbbm{1}$ is denoted by $\Gamma_\alpha^\gamma$ in \cite{CDFMT-2012}, and what we denote here by $\langle\xi,(T_\lambda^{(\ell)}+\alpha\mathbbm{1})\xi\rangle$ is precisely the quantity $\Phi^\lambda_\alpha[\xi]$ in \cite{CDFMT-2012}. Observe that despite the fact that in \cite[Eqns.~(2-29)-(2.30)]{CDFMT-2012} the charges of the functions in $\mathcal{D}(H_\alpha)$ are qualified by the boundary condition (\textsc{bc1}) under the constraint that $(T_\lambda+\alpha\mathbbm{1})\xi\equiv\Gamma_\alpha^\gamma\xi\in L^2(\mathbb{R}^3)$, it is clear that (\textsc{bc1}) is an identity among $H^{\frac{1}{2}}$-functions: our present formulation \eqref{eq:DHab} does not involve the unnecessary space $L^2(\mathbb{R}^3)$ and formulates the constraint on $\xi$ already in the form $(T_\lambda+\alpha\mathbbm{1})\xi\in H^{\frac{1}{2}}(\mathbb{R}^3)$.

As mentioned already (see Remark \ref{rem:remark_history} above), the possibility of the existence of self-adjoint extensions of $\mathring{H}$ of TMS type, \emph{other than $H_\alpha$}, was known since long and investigated in several studies, however, the precise role of the space of charges in the extension scheme was missed. Only in a recent work of one of us in collaboration with Correggi, Dell'Antonio, Finco, and Teta \cite{CDFMT-2015}, again through a quadratic form approach, a substantial progress was finally made by characterising the family of extensions of the form $H^{[\beta]}_{\alpha=0}$, namely those relative to the (2+1)-model \emph{at unitarity}, namely, with infinite scattering length $\alpha^{-1}=\infty$. Our Theorem \ref{thm:DHab_form} is immediately comparable with \cite[Theorem 2.1]{CDFMT-2015}, and Theorem \ref{thm:DHab} with \cite[Theorem 2.2]{CDFMT-2015}. Because of the restriction to $\alpha=0$ it was convenient in \cite{CDFMT-2015} to use homogeneous Sobolev spaces, which essentially corresponds to sending $\lambda\to 0$ in our formulas.

In the remaining part of this Section we shall prove Theorems \ref{thm:DHab} and \ref{thm:DHab_form}, whereas in the next Section we discuss the occurrence of the Ter-Martirosyan-Skornyakov asymptotics.

It is convenient to prove first the characterisation of the quadratic form of $H_{\alpha}^{[\beta]}$.

\begin{proof}[Proof of Theorem \ref{thm:DHab_form}]
Formula \eqref{eq:def_Halpha_beta_qform} indicates that $g\in\mathcal{D}[H_{\alpha}^{[\beta]}]$ is of the form $g=F+u_\xi$ where $F \in H^1_f(\mathbb{R}^3\times\mathbb{R}^3)$ and $\xi\in\mathcal{D}[\mathcal{A}_{\lambda,\alpha}^{[\beta]}]$. From \eqref{eq:ext_param_A_lab} we see that $\xi\in\mathcal{D}[\mathcal{A}_{\lambda,\alpha}^{[\beta]}]$ reads $\xi=\bigoplus_{\ell=0}^\infty\xi^{(\ell)}$, where $\xi^{(\ell)}\in\mathcal{D}[\mathcal{A}_{\lambda,\alpha}^{(\ell)}]=H^{1/2}_{\ell}(\mathbb{R}^3)$ for $\ell\neq 1$ (owing to \eqref{eq:DformA} and \eqref{eq:Al0}), and $\xi^{(\ell=1)}\in H^{1/2}_{\ell=1}(\mathbb{R}^3)\dotplus\mathrm{span}\{\Xi_{1,-1},\Xi_{1,0},\Xi_{1,1}\}$ (owing to \eqref{eq:quadratic_forms_Abeta_v2}). This proves the expression of $\mathcal{D}[H_{\alpha}^{[\beta]}]$ given in \eqref{eq:DHab_form}. Formula \eqref{eq:def_Halpha_beta_qform} also indicates that 
$(H_{\alpha}^{[\beta]}+\lambda\mathbbm{1})[F+u_\xi]=(\mathring{H}_F+\lambda\mathbbm{1})[F]+\mathcal{A}_{\lambda,\alpha}^{[\beta]}[\xi]$, equivalently,
\[
H_{\alpha}^{[\beta]}[F+u_\xi]\;=\;\lambda\|F\|_\cH^2-\lambda\|F+u_\xi\|_\cH^2+H_\mathrm{free}[F]+\mathcal{A}_{\lambda,\alpha}^{[\beta]}[\xi]\,.
\]
In turn, $\xi^{(\ell)}\equiv\xi^{(\ell)}_\mathrm{reg}$ and $\mathcal{A}_{\lambda,\alpha}^{[\beta]}[\xi^{(\ell)}]=\mathcal{A}_{\lambda,\alpha}^{(\ell)}[\xi^{(\ell)}_\mathrm{reg}]=2\,\langle\xi^{(\ell)}_\mathrm{reg},(T_\lambda^{(1)}+\alpha\mathbbm{1})\xi^{(\ell)}_\mathrm{reg}\rangle_{H^{\frac{1}{2}},H^{-\frac{1}{2}}}$ for $\ell\neq 1$ (owing to \eqref{eq:DformA} and \eqref{eq:Al0}), whereas $\xi^{(\ell=1)}=\xi^{(\ell=1)}_\mathrm{reg}+\sum_{n=-1}^1 q_n\Xi_{1,n}$ and 
$\mathcal{A}_{\lambda,\alpha}^{[\beta]}[\xi^{(\ell=1)}]=\mathcal{A}_{\lambda,\alpha}^{(\ell=1)}[\xi^{(\ell=1)}]=2\,\langle\xi^{(\ell=1)}_\mathrm{reg},(T_\lambda^{(1)}+\alpha\mathbbm{1})\xi^{(\ell=1)}_\mathrm{reg}\rangle_{H^{\frac{1}{2}},H^{-\frac{1}{2}}}+2\sum_{n=-1}^1 \beta_n|q_n|^2$ (owing to \eqref{eq:quadratic_forms_Abeta_v2}). This proves the expression of $H_{\alpha}^{[\beta]}[F+u_\xi]$ given in \eqref{eq:DHab_form}.
\end{proof}

We pass now to the proof of the characterisation of the operator $H_{\alpha}^{[\beta]}$.

\begin{proof}[Proof of Theorem \ref{thm:DHab}] In view of formula \eqref{eq:def_Halpha_beta}, each $g\in\mathcal{D}(H_{\alpha}^{[\beta]})$ decomposes as $g=F+u_\xi$ where 
\[
F\;:=\;f+(\mathring{H}_F+\lambda\mathbbm{1})^{-1}\big(u_{\mathcal{A}_{\lambda,\alpha}^{[\beta]}\xi}\big)\;\in\;\mathcal{D}[\mathring{H}_F]\;=\;H^2_f(\mathbb{R}^3\times\mathbb{R}^3)
\]
and $\xi\in\mathcal{D}(\mathcal{A}_{\lambda,\alpha}^{[\beta]})\subset\ker(\mathring{H}^*+\lambda\mathbbm{1})$. Since $H_{\alpha}^{[\beta]}+\lambda\mathbbm{1}$ acts on $g$ as $\mathring{H}+\lambda\mathbbm{1}$, formula \eqref{eq:action_DHab} follows at once. %Back to the proof of formula \eqref{eq:DHab}

Because of \eqref{eq:ext_param_A_lab}, $\xi=\bigoplus_{\ell=0}^\infty\xi^{(\ell)}$ where, for $\ell\neq 1$, $\xi^{(\ell)}\equiv\xi^{(\ell)}_{\mathrm{reg}}\in\mathcal{D}(\mathcal{A}_{\lambda,\alpha}^{(\ell),F})$ and hence (owing to \eqref{eq:DFriedrA})
$\xi^{(\ell)}\in H^{\frac{1}{2}}_{\ell}(\mathbb{R}^3)$ with $(T^{(\ell)}_\lambda+\alpha)\xi^{(\ell)}\in H^{\frac{1}{2}}(\mathbb{R}^3)$, whereas, for $\ell=1$, $\xi^{(1)}=\xi_{\mathrm{reg}}^{(1)}+\sum_{n=-1}^1 q_n\Xi_{1,n}$ with $\xi_{\mathrm{reg}}^{(1)}\in\mathcal{D}(\mathcal{A}_{\lambda,\alpha}^{(\ell=1),F})$ and hence (owing to \eqref{eq:DFriedrA} again) $\xi_{\mathrm{reg}}^{(1)}\in H^{\frac{1}{2}}_{\ell=1}(\mathbb{R}^3)$ with $(T^{(1)}_\lambda+\alpha)\xi_{\mathrm{reg}}^{(1)}\in H^{\frac{1}{2}}(\mathbb{R}^3)$. Thus, setting $\xi_{\mathrm{reg}}:=\bigoplus_{\ell=0}^\infty\xi^{(\ell)}_{\mathrm{reg}}$, the decomposition of the charge $\xi$ stated in the first part of formula \eqref{eq:DHab} is proved.

We pass now to the proof of the two boundary conditions (\textsc{bc}1) and (\textsc{bc}2) of \eqref{eq:DHab}. 
For $g\in\mathcal{D}(H_{\alpha}^{[\beta]})$, one has
\[
(H_{\alpha}^{[\beta]}+\lambda\mathbbm{1})[g',g]\;=\;\langle g',(H_{\alpha}^{[\beta]}+\lambda\mathbbm{1})g\rangle_\cH\qquad\forall g'\in\mathcal{D}[H_{\alpha}^{[\beta]}]\,.
\]
The identity above must be true in particular for all $g'$ in the form domain \eqref{eq:DHab_form}  of the form $g'=F'+u_{\xi'}$ with $F'\in H^1_f(\mathbb{R}^3\times\mathbb{R}^3)$ and $\xi'\in H^{\frac{1}{2}}(\mathbb{R}^3)$, namely, those functions $g'$ of the form domain with only regular charges. With this choice, and decomposing $g=F+u_\xi$, $\xi=\xi_{\mathrm{reg}}+\sum_{n=-1}^1 q_n\Xi_{1,n}$, \eqref{eq:DHab_form} yields
\[
\begin{split}
(H_{\alpha}^{[\beta]}+&\lambda\mathbbm{1})[g',g]\;=\;(H_{\mathrm{free}}+\lambda\mathbbm{1})[F',F]+\mathcal{A}_{\lambda,\alpha}^{[\beta]}[\xi',\xi] \\
&=\;(H_{\mathrm{free}}+\lambda\mathbbm{1})[F',F]+2\,\int_{\mathbb{R}^3}\overline{\,\widehat{\xi'}(p)}\,\big(\widehat{(T_\lambda\xi_{\mathrm{reg}})}(p)+\alpha\,\widehat{\xi_{\mathrm{reg}}}(p)\big)\,\ud p
\end{split}
\]
(where it was crucial that the charge $\xi'$ has no singular part), whereas \eqref{eq:action_DHab} yields
\[
\begin{split}
\langle g',(H_{\alpha}^{[\beta]}+\lambda\mathbbm{1})g\rangle_\cH\;&=\;\langle F',(H_{\mathrm{free}}+\lambda\mathbbm{1})F\rangle_\cH+\langle u_{\xi'}, (H_{\mathrm{free}}+\lambda\mathbbm{1})F \rangle_{\cH}\\
&=\;(H_{\mathrm{free}}+\lambda\mathbbm{1})[F',F]+\langle u_{\xi'}, (H_{\mathrm{free}}+\lambda\mathbbm{1})F \rangle_{\cH}\,.
\end{split}
\]
Observe also that
\[
\begin{split}
 \langle u_{\xi'}, (H_{\mathrm{free}}+\lambda\mathbbm{1})F \rangle_{\cH}\;&=\;\iint\frac{\overline{\,\widehat{\xi'}(p_1)}-\overline{\,\widehat{\xi'}(p_2)}}{p_1^2+p_2^2+\mu p_1\cdot p_2+\lambda}\,\big((H_{\mathrm{free}}+\lambda\mathbbm{1})F\big)^{\widehat{\;}}(p_1,p_2)\,\ud p_1\ud p_2 \\
 &=\;2\int_{\mathbb{R}^3}\!\ud p_1\overline{\,\widehat{\xi'}(p_1)}\,\Big(\int_{\mathbb{R}^3}\!\ud p_2\,\widehat{F}(p_1,p_2)\Big)
\end{split}
\]
(since $F\in H^2_f(\mathbb{R}^3\times\mathbb{R}^3)$, the inverse Fourier transform of the function $p\mapsto \int_{\mathbb{R}^3}\!\ud q\,\widehat{F}(p,q)$ is a function in $H^{\frac{1}{2}}(\mathbb{R}^3)$, owing to a standard trace theorem, e.g., \cite[Lemma 16.1]{Trtar-SobSpaces_interp}).
Therefore, $\forall\xi'\in H^{\frac{1}{2}}(\mathbb{R}^3)$,
\[
\int_{\mathbb{R}^3}\overline{\,\widehat{\xi'}(p)}\,\big(\widehat{(T_\lambda\xi_{\mathrm{reg}})}(p)+\alpha\,\widehat{\xi_{\mathrm{reg}}}(p)\big)\,\ud p\;=\;\int_{\mathbb{R}^3}\!\ud p_1\overline{\,\widehat{\xi'}(p_1)}\,\Big(\int_{\mathbb{R}^3}\!\ud p_2\,\widehat{F}(p_1,p_2)\Big)\,,
\]
which by density gives condition (\textsc{bc}1) as an identity in $H^{\frac{1}{2}}(\mathbb{R}^3)$.

If the charge $\xi$ of $g\in\mathcal{D}(H_{\alpha}^{[\beta]})$ has a (non-zero) component with symmetry $\ell=1$, say, $\xi^{(1)}=\xi_{\mathrm{reg}}^{(1)}+\sum_{n=-1}^1 q_n\Xi_{1,n}$, it was already established in \eqref{eq:properties_of_xi_in_DAb1} that the regular part $\xi_{\mathrm{reg}}^{(1)}$ of $\xi^{(1)}$ satisfies the condition (\textsc{bc}2). Therefore, for a generic charge $\xi=\xi_{\mathrm{reg}}+\sum_{n=-1}^1q_n\Xi_{1,n}$, $\xi_{\mathrm{reg}}=\bigoplus_{\ell=0}^\infty\xi_{\mathrm{reg}}^{(\ell)}$, by orthogonality one has
\[
\begin{split}
 \int_{\mathbb{R}^3}\overline{\,\widehat{\Xi}_{1,n}(p)}\,\big(\widehat{(T_\lambda\xi_{\mathrm{reg}})}(p)+\alpha\,\widehat{\xi_{\mathrm{reg}}}(p)\big)\,\ud p\;&=\; \int_{\mathbb{R}^3}\overline{\,\widehat{\Xi}_{1,n}(p)}\,\big(\widehat{(T_\lambda\xi^{(1)}_{\mathrm{reg}})}(p)+\alpha\,\widehat{\xi^{(1)}_{\mathrm{reg}}}(p)\big)\,\ud p\\
 &=\;\beta_n|q_n|^2\,,
\end{split}
\]
which completes the proof of (\textsc{bc}2).
% 
% 
% (For all other components $\xi^{(\ell)}$ of $\xi$ with $\ell\neq 1$ the l.h.s.~of (\textsc{bc}2) would be clearly zero, due to orthogonality.)
\end{proof}

\section{Emergence of the Ter-Martirosyan--Skornyakov asymptotics}\label{sec:emergence_of_TMS}

In this Section we elaborate on the emergence of the Ter-Martirosyan--Skornyakov asymptotics for functions $g\in\mathcal{D}(H_{\alpha}^{[\beta]})$.

In fact, the boundary condition (\textsc{bc1}) in \eqref{eq:DHab} \emph{is} a TMS condition. For the precise meaning of that, we need to recall the following simple facts.

\begin{lemma}{Lemma}\label{lem:link_xireg_eta}
Under the assumptions of Theorem \ref{thm:DHab}, if $g\in\mathcal{D}(H_{\alpha}^{[\beta]})$, then its regular part $F$ and the charges $\xi$, $\xi_\mathrm{reg}$, and $\eta=\mathcal{A}_{\lambda,\alpha}^{[\beta]}\xi$ defined by the decompositions \eqref{eq:DHab} and \eqref{eq:decompositionD_Hdostar_2+1} of $g$, satisfy
\begin{equation}\label{eq:asympt_reg_part}
\int_{\mathbb{R}^3}\widehat{F}(p_1,p_2)\,\ud p_2\;=\;{\textstyle\frac{1}{2}}\,\widehat{(W_\lambda\eta)}(p_1)
\end{equation}
and
 \begin{equation}\label{eq:link_xireg_eta}
 T_\lambda\xi_\mathrm{reg}+\alpha\xi_\mathrm{reg}\;=\;{\textstyle\frac{1}{2}}\,W_\lambda\eta
 \end{equation}
as identities in $H^{\frac{1}{2}}(\mathbb{R}^3)$, and 
\begin{equation}\label{eq:Axixireg}
 \mathcal{A}_{\lambda,\alpha}^{[\beta]}\xi\;=\;2\,W_\lambda^{-1}(T_\lambda+\alpha\mathbbm{1})\,\xi_{\mathrm{reg}}
\end{equation}
as an identity in $H^{-\frac{1}{2}}(\mathbb{R}^3)$.
\end{lemma}

\begin{proof}
In fact, we know already \eqref{eq:Axixireg} from \eqref{eq:DFriedrA} and \eqref{eq:xi_reg}, because on each $\ell$-th sector 
$\mathcal{A}_{\lambda,\alpha}^{(\ell),\beta}\,\xi=\mathcal{A}_{\lambda,\alpha}^{(\ell),F}\,\xi_{\mathrm{reg}}=2\,W_\lambda^{-1}(T_\lambda^{(\ell)}+\alpha\mathbbm{1})\,\xi_{\mathrm{reg}}$, but we can also argue as follows. From $g=F+u_\xi$,
\[
\begin{split}
\int_{|p_2|< R}\widehat{F}(p_1,p_2)\,\ud p_2\;&=\;\int_{|p_2|< R}\widehat{g}(p_1,p_2)\,\ud p_2-\int_{|p_2|< R}\widehat{u_\xi}(p_1,p_2)\,\ud p_2\,,
\end{split}
\]
and from \eqref{eq:g_*_asymptotics_2+1}, the proof of which includes also the asymptotics for $u_\xi$ (see \cite[Lemma 4]{MO-2016}), one has 
\[
\begin{split}
\int_{|p_2|< R}\widehat{g}(p_1,p_2)\,\ud p_2\;&=\;4\pi\widehat{\xi}(p_1) R+\big(-\widehat{(T_\lambda\,\xi)}(p_1)+{\textstyle\frac{1}{2}}\,\widehat{(W_\lambda\,\eta)}(p_1)\big)+o(1)\\
\int_{|p_2|< R}\widehat{u_\xi}(p_1,p_2)\,\ud p_2\;&=\;4\pi\widehat{\xi}(p_1) R-\widehat{(T_\lambda\,\xi)}(p_1)+o(1)
\end{split}
\]
as $R\to +\infty$, which gives \eqref{eq:asympt_reg_part}. Combining \eqref{eq:asympt_reg_part} with (\textsc{bc1}) in \eqref{eq:DHab} yields \eqref{eq:link_xireg_eta}. Last, \eqref{eq:Axixireg} follows because $\eta=\mathcal{A}_{\lambda,\alpha}^{[\beta]}\xi$ and $W_\lambda$ is a bijection $H^{-\frac{1}{2}}(\mathbb{R}^3)\to H^{\frac{1}{2}}(\mathbb{R}^3)$ (Proposition \ref{prop:T-W}(i)).
\end{proof}

%%%%%%%%%%%%%%%%%%%%% FROM HERE

Plugging $\eta=\mathcal{A}_{\lambda,\alpha}^{[\beta]}\xi$ given by \eqref{eq:Axixireg} into \eqref{eq:g_*_asymptotics_2+1} gives
\begin{equation}\label{eq:trueTMS_general}
\begin{split}
\int_{|p_2|< R}\widehat{g}(p_1,p_2)\,\ud p_2\;&=\;4\pi\widehat{\xi}(p_1) R+\alpha\,\widehat{\xi}_{\mathrm{reg}}(p_1)+o(1)\quad\textrm{ as }R\to +\infty\,.
\end{split}
\end{equation}
For those $g$'s in $\mathcal{D}(H_{\alpha}^{[\beta]})$ whose charge $\xi$ has only \emph{regular} part the asymptotics take the genuine TMS form
\begin{equation}\label{eq:trueTMS_general_reg}
\begin{split}
\int_{|p_2|< R}\widehat{g}(p_1,p_2)\,\ud p_2\;&=\;(4\pi R+\alpha)\,\widehat{\xi}_{\mathrm{reg}}(p_1)+o(1)\quad\textrm{ as }R\to +\infty\,.
\end{split}
\end{equation}
In \eqref{eq:trueTMS_general_reg} the point-wise approximation of the function $\int_{|p_2|< R}\widehat{g}(p_1,p_2)\,\ud p_2$ is given by a $H^{\frac{1}{2}}$-function (because so is $\xi_{\mathrm{reg}}$); in \eqref{eq:trueTMS_general}, due to the possible presence of more singular charges in $\xi=\xi_\mathrm{reg}+\sum_{n=-1}^1 q_n\Xi_{1,n}$, the point-wise approximating function for  $\int_{|p_2|< R}\widehat{g}(p_1,p_2)\,\ud p_2$ is in turn less regular.

Thus, the Hamiltonian $H_\alpha\equiv H_\alpha^{[\beta=\infty]}$ discussed in Section \ref{sec:reconstr_Halpha} and in the works \cite{Finco-Teta-2012,CDFMT-2012} has a genuine TMS condition for the functions of its domain, whereas for the Hamiltonians $H_\alpha\equiv H_\alpha^{[\beta]}$, due to the occurrence of more singular charges, the TMS condition survives only for sufficiently regular elements of the domain.

\section{Multiplicity of extensions and three-body boundary condition}\label{sec:3bodycond}

In this Section we discuss the interpretation of the boundary condition $\textsc{(bc2)}$ in \eqref{eq:DHab}, that is, the condition that qualifies the link between the regular and the singular part of the charge associated with an element in the domain of the self-adjoint Hamiltonian $H_\alpha^{[\beta]}$, provided of course that the charge has a non-zero component in the sector of symmetry $\ell=1$.

In fact, $\textsc{(bc2)}$ has eventually a natural interpretation of \emph{additional three-body boundary condition} (besides the two-body Ter-Martirosyan--Skornyakov boundary condition at the coincidence hyperplanes).

The circumstance that for three-body models with point interaction there is a regime of the mass parameter ($m\in(m^*,m^{**})$ in our notation for the present (2+1)-fermionic model) in which the Hamiltonian is not fully qualified by the sole specification of a boundary condition of TMS type, when two particles come on top of each other, and instead one needs to fix a further three-body parameter for an unambiguous realisation of the Hamiltonian, emerged quite immediately in the study of point interactions, even at a non-rigorous level. 
Indeed, right after Skornyakov and Ter-Martirosyan had defined their (formal) three-body model for identical bosons \cite{TMS-1956}, it was noted by Danilov \cite{Danilov-1961} that the equation for eigenvalues has a non-physical (today we would say: incompatible with self-adjointness) continuum of negative solutions. To resolve such a continuum, an additional constraint was prescribed, in the form of a real constant of proportionality between the coefficients of the two \emph{leading} singularities of the corresponding eigenfunctions
(see, e.g., \cite[Eq.~(3)]{Minlos-Faddeev-1961-2}), thus resulting in a discrete spectrum for each value of the additional parameter. This real constant was given in some vague sense the meaning of a three-body parameter, and it was indeed its emergence to motivate the subsequent first attempt of Minlos and Faddeev \cite{Minlos-Faddeev-1961-1,Minlos-Faddeev-1961-2} for a rigorous explanation in terms of a one-parameter family of self-adjoint realisations of the formal Hamiltonian, a picture in which Danilov's parameter is seen on the same footing as our $\beta$ in \eqref{eq:DHab}.

In the present setting, the role of $\beta\equiv(\beta_{-1},\beta_0,\beta_1)$, $\beta_j\in\mathbb{R}\cup\{\infty\}$, as a parametrisation of each self-adjoint extension $H_\alpha^{[\beta]}$ was rigorously established within the Kre{\u\i}n-Vi\v{s}ik-Birman extension theory. We want now to comment about its interpretation as the parameter for a three-body condition.

For $g\in\mathcal{D}(H_\alpha^{[\beta]})$ let assume non-restrictively that the corresponding charge $\xi$ in \eqref{eq:DHab} belongs entirely to the sector of symmetry $\ell=1$. First we observe that if we write (from \eqref{eq:A-extensions_op})
\begin{equation}\label{eq:xireg_in_components}
\begin{split}
&\xi\;=\;\xi_{\mathrm{reg}}+\xi_{\mathrm{sing}}\,,\qquad \xi_{\mathrm{reg}}\;=\;\xi_1+\xi_2^{(\beta)}\,,\qquad\xi_{\mathrm{sing}}\;=\sum_{n=-1}^1 q_n\Xi_{1,n}\,, \\
\xi_1&\in\mathcal{D}(\overline{\mathcal{A}_{\lambda,\alpha}^{(1)}})\,,\qquad \xi_2^{(\beta)}\;:=\;(\mathcal{A}_{\lambda,\alpha}^{(1),F})^{-1}\xi_{\mathrm{sing}}^{(\beta)}\,,\qquad \xi_{\mathrm{sing}}^{(\beta)}\;:=\sum_{n=-1}^1 \beta_n q_n\Xi_{1,n}\,,
\end{split}
\end{equation}
then the boundary condition $\textsc{(bc2)}$ in \eqref{eq:DHab} is only effective as a constraint between $\xi_2^{(\beta)}$ and $\xi_{\mathrm{sing}}$, and insensitive to $\xi_1$, for
\[
\langle \Xi_{1,n},(T_\lambda^{(1)}+\alpha\mathbbm{1})\xi_1\rangle_{H^{-\frac{1}{2}},H^{\frac{1}{2}}}\;=\;\langle \Xi_{1,n},\overline{\mathcal{A}_{\lambda,\alpha}^{(1)}}\,\xi_1\rangle_{H^{-1/2}_{W_\lambda}}\;=\;\langle \mathcal{A}_{\lambda,\alpha}^{(1)\star}\Xi_{1,n},\xi_1\rangle_{H^{-1/2}_{W_\lambda}}\;=\;0\,.
\]
Thus, $\textsc{(bc2)}$ is equivalent to
\begin{equation}
\langle \Xi_{1,n},(T_\lambda^{(1)}+\alpha\mathbbm{1})\,\xi_2^{(\beta)}\rangle_{H^{-\frac{1}{2}},H^{\frac{1}{2}}}\;=\;\beta_n|q_n|^2\qquad (\textsc{bc2})'\,.
\end{equation}

In \eqref{eq:xireg_in_components} we can think of $\xi=\xi_{\mathrm{reg}}+\xi_{\mathrm{sing}}=\xi_1+\xi_2^{(\beta)}+\xi_{\mathrm{sing}}$ as the sum of a `very regular' component $\xi_1$, a `less regular' $\xi_2^{(\beta)}$, and a singular $\xi_{\mathrm{sing}}$. In other words, $\xi_{\mathrm{sing}}$ and $\xi_2^{(\beta)}$ represent, respectively, the leading and the next-to-leading singularity in the charge $\xi$, and the parameter $\beta$ in $(\textsc{bc2})'$ prescribes a constraint between such singularities, equivalent to the $\beta$-dependent constraint between $\xi_{\mathrm{sing}}$ and $\xi_2^{(\beta)}$ given by \eqref{eq:xireg_in_components}.

It can be argued (it is part of our main conjecture and we shall elaborate more on that in Section \ref{sec:mass_thr_and_evidence_of_conj}) that the singularity in $\xi_{\mathrm{sing}}$ has (radially) the form
\begin{equation}\label{eq:xising_asympt}
 \widehat{\xi}_{\mathrm{sing}}(p)\;\sim\;|p|^{-2+s(m)}\qquad\textrm{as}\qquad|p|\to+\infty
\end{equation}
for some mass-dependent exponent $s(m)\in(0,1)$ -- observe that $s(m)<1$ ensures that $\xi_{\mathrm{sing}}\in H^{-\frac{1}{2}}(\mathbb{R}^3)$, as it has to be because $\xi_{\mathrm{sing}}\in\ker(\mathcal{A}_{\lambda,\alpha}^{(\ell)})^\star\subset H^{-1/2}_{\ell=1}(\mathbb{R}^3)$, whereas $s(m)> 0$ makes $\xi_{\mathrm{sing}}$ a charge in  $H^{s}(\mathbb{R}^3)$ with $s<\frac{1}{2}$. One can also argue what the milder singularity is in $\xi_2^{(\beta)}=(\mathcal{A}_{\lambda,\alpha}^{(1),F})^{-1}\xi_{\mathrm{sing}}^{(\beta)}$, namely what the regularisation effect is of the operator $(\mathcal{A}_{\lambda,\alpha}^{(1),F})^{-1}$ on the function $\xi_{\mathrm{sing}}^{(\beta)}$. From \eqref{eq:xireg_in_components}, in each sub-sector of symmetry $n=-1,0,1$ one has $\widehat{\xi}_{\mathrm{sing},n}^{(\beta)}\sim\beta_n  |p|^{-2+s(m)}$, that is, $\xi_{\mathrm{sing}}^{(\beta)}$
has by construction the same singularity as $\xi_{\mathrm{sing}}$. With heuristics based on the same arguments as in \cite[Proposition 2.2 and Eq.~(3.53)]{CDFMT-2015} we find
\begin{equation}\label{eq:xi2_asympt}
 \widehat{\xi_{2,n}^{(\beta)}}(p)\;\sim\;\beta_n\,|p|^{-2-s(m)}\qquad\textrm{as}\qquad|p|\to+\infty\,.
\end{equation}

Summarising, we do expect for the charge $\xi$ the asymptotics (see, e.g., \cite[Eq.~(2.17)]{CDFMT-2015}
\begin{equation}\label{eq:xi_asympt}
\widehat{\xi}(p)\;\sim\;\sum_{n=-1}^1\Big(\frac{1}{\;|p|^{2-s(m)}}+\frac{\beta_n A_n}{\;|p|^{2+s(m)}}+o(|p|^{-2-s(m)})\Big)\,,\qquad |p|\to+\infty\,,
\end{equation}
for some coefficients $A_n\in\mathbb{C}$, where the condition $(\textsc{bc2})'$ is implemented via its equivalent form $\xi_2^{(\beta)}=(\mathcal{A}_{\lambda,\alpha}^{(1),F})^{-1}\sum_{n=-1}^1 \beta_n q_n\Xi_{1,n}$, in view of \eqref{eq:xireg_in_components}, \eqref{eq:xising_asympt}, and \eqref{eq:xi2_asympt} above.

In turn, the singularity in $\xi$ determines the singularity of the function $g\in\mathcal{D}(H_\alpha^{[\beta]})$, more precisely of its singular (i.e., non-$H^2$) component $u_\xi$, by means of definition \eqref{eq:u_xi}. The correspondence $\xi\mapsto u_\xi$ is linear, thus resulting in a term-by-term counterpart of the asymptotics \eqref{eq:xi_asympt} for $\widehat{u}_\xi(p_1,p_2)$ as $|p_1|,|p_2|\to\infty$. For a cleaner inspection of the latter, it is convenient to consider the function $u'_\xi$ corresponding to a homogeneous version of $u_\xi$, defined by
\begin{equation}\label{eq:u_xi'}
 \widehat{u}'_\xi(p_1,p_2)\;:=\;\frac{\widehat{\xi}(p_1)-\widehat{\xi}(p_2)}{p_1^2+p_2^2+\mu\,p_1\cdot p_2}\,,\qquad p_1,p_2\in\mathbb{R}^3\,,
\end{equation}
which is tantamount as subtracting a very regular function from $u_\xi$. Now formulas \eqref{eq:xi_asympt} and \eqref{eq:u_xi'} allow to perform an easy scaling argument, that can be also translated back to position coordinates via inverse Fourier transform. The result, as found and discussed in \cite[Remark 2.8]{CDFMT-2015}, is
\begin{equation}\label{eq:triple_coinc_pt}
 u'_\xi(r\widehat{x}_1,r\widehat{x}_2)\;\underset{r\to 0}{\sim}\;\frac{1}{\:r^2}\sum_{n=-1}^1\Big(\frac{1}{\;r^{s(m)}}+\beta_n a_n\,r^{s(m)}+o(r^{s(m)})\Big)
\end{equation}
(here $\widehat{x}_j\in\mathbb{R}^3$, $|\widehat{x}_j|=1$, $j=1,2$).

Formula \eqref{eq:triple_coinc_pt} expresses an \emph{asymptotics at the triple coincidence point}. Thus, the boundary condition $\textsc{(bc2)}$ in \eqref{eq:DHab} translates into the precise $\beta$-dependent asymptotics \eqref{eq:triple_coinc_pt} that prescribes the proportionality relation between the constants of the leading and of the next-to-leading singularity of the three-body wave function when all three particles come on top of each other in a \emph{collapse onto the centre of mass}.

Also, the explicit form of \eqref{eq:triple_coinc_pt} matches precisely the analogous findings in the physical literature -- see, e.g., \cite[note 43]{Werner-Castin-2006-PRA}, and more recently \cite{Kartavtsev-Malykh-2016,Kartavtsev-Malykh-2016-proc}.

In conclusion: assigning solely the two-body $(|x_j|^{-1}+\alpha)$-singularity (as $|x_j|\to 0$) at the contact between two particles is not enough in general to qualify a domain of self-adjointness: for $m\in(m^*,m^{**})$, next to the two-body parameter $\alpha$ (scattering length), an additional parameter $\beta\in\mathbb{R}\cup\{\infty\}$ need be specified, which regulates the three-body singularity at the simultaneous spatial coincidence of all three particles in terms of the asymptotics \eqref{eq:triple_coinc_pt}.

\section{Mass thresholds and evidence of our conjecture}\label{sec:mass_thr_and_evidence_of_conj}

In this Section we discuss an amount of stringent evidence corroborating our main conjecture and the origin of the mass threshold $m^{**}$. Some of the following considerations originate in fact from the analysis of the previous literature, but only after having developed here the correct extension framework a la Kre{\u\i}n, Vi\v{s}ik, and Birman can we finally formulate the appropriate arguments.

In a series of recent notable works \cite{Minlos-2012-preprint_30sett2011,Minlos-2012-preprint_1nov2012,Minlos-RusMathSurv-2014}, Minlos established a result that is particularly relevant for the present discussion: we re-phrase it here within our current notation and together with an auxiliary clarifying result found in \cite{CDFMT-2015}.

\begin{theorem}{Proposition}\label{prop:critical_masses}

%[\cite{Minlos-2012-preprint_30sett2011,Minlos-2012-preprint_1nov2012,Minlos-RusMathSurv-2014} and $\textrm{\cite[Appendix A]{CDFMT-2015}}$] 
\begin{itemize}
 \item[(i)] \emph{\cite[Appendix A]{CDFMT-2015}} -- For $m>0$ and $s\in[0,1]$ the integral equation
 \begin{equation}\label{eq:s-integral-equation}
\pi\sqrt{\frac{\,m(m+2)}{(m+1)^2}}+\int_{-1}^1\!\ud y\,y\int_0^{+\infty}\!\!\ud r\,\frac{r^s}{\,r^2+1+\frac{2}{m+1}ry}\;=\;0
\end{equation}
 defines a continuous and monotone increasing function $s\mapsto m(s)$ with $m(0)=m^*\approx (13.607)^{-1}$ given by 
 $\Lambda(m^*)=1$ in \eqref{eq:Lambdam} and $m(1)=m^{**}\approx (8.62)^{-1}$ given by \eqref{eq:m**root}. The integral equation \eqref{eq:s-integral-equation} defines also a continuous and monotone increasing function $m\mapsto m(s)$, the inverse of $s\mapsto m(s)$.
 \item[(ii)] \emph{\cite{Minlos-2012-preprint_30sett2011,Minlos-2012-preprint_1nov2012,Minlos-RusMathSurv-2014}} -- For given $m>m^*$ and $\lambda>0$ consider $T_{\lambda}^{(1)}$ defined by \eqref{eq:Tlambda} and \eqref{eq:T-Tell} as an operator on $L^2_{\ell=1}(\mathbb{R}^3)$ with domain $\mathcal{D}(T_{\lambda}^{(1)})=H^{3/2}_{\ell=1}(\mathbb{R}^3)$, hence densely defined and symmetric in $L^2_{\ell=1}(\mathbb{R}^3)$ (Proposition \ref{prop:T-W}(ii)). Let $m^{**}_{\textsc{m}}=m(\frac{1}{2})\approx(12.315)^{-1}$ be the unique root of the integral equation \eqref{eq:s-integral-equation} when $s=\frac{1}{2}$. Then for $m\geqslant m^{**}_{\textsc{m}}$ the operator $T_{\lambda}^{(1)}$ is essentially self-adjoint and for $m\in(m^*,m^{**}_{\textsc{m}})$ it has deficiency index equal to three. When $m\in(m^*,m^{**}_{\textsc{m}})$, the kernel of $T_{\lambda=0}^{(1)}$ is spanned by the three functions $\Xi'_{1,-1}$, $\Xi'_{1,0}$, and $\Xi'_{1,1}$ given by
 \begin{equation}
 \widehat{\Xi}'_{1,n}(p)\;=\;\frac{\mathbf{1}_{\{|p|\geqslant 1\}}}{\,|p|^{2-s(m)\,}}Y_{1,n}(\Omega_p)\qquad n=-1,0,1
 \end{equation}
 in polar coordinates $p\equiv|p|\Omega_p$.
\end{itemize}
\end{theorem}

Given now $\alpha\in\mathbb{R}$, and taking $\lambda$ large enough (depending on $\alpha$), the operator $T_{\lambda}^{(1)}+\alpha\mathbbm{1}$ is positive in $L^2_{\ell=1}(\mathbb{R}^3)$ with strictly positive bottom, as discussed in the proof of Proposition \ref{prop:Apos} and in \cite[Proposition 3.1]{CDFMT-2012}. Therefore, an immediate corollary of Proposition \ref{prop:critical_masses}(ii) is:

\begin{lemma}{Corollary}\label{cor:from_Minlos}
Given $\alpha\in\mathbb{R}$, and taking $\lambda$ large enough, one has
\begin{equation}
\dim\ker(T_{\lambda}^{(1)*}+\alpha\mathbbm{1})\;=\;\begin{cases}
 \;3 & m\in(m^*,m^{**}_{\textsc{m}}) \\
 \;0 & m\geqslant m^{**}_{\textsc{m}}\,.
\end{cases}
\end{equation}
\end{lemma}

Minlos's analysis was motivated by the (wrong) belief that the self-adjoint extensions of $T_{\lambda}^{(1)}$ in $L^2_{\ell=1}(\mathbb{R}^3)$ label the self-adjoint realisations of the Ter-Martirosyan--Skornyakov Hamiltonian (Remark \ref{rem:remark_history}). However, it has the virtue of showing that $T_{\lambda}^{(1)*}$ -- the adjoint  in $L^2_{\ell=1}(\mathbb{R}^3)$ -- has \emph{the same formal action} as $T_{\lambda}^{(1)}$ \cite[Eq.~(4.14)-(4.16)]{Minlos-2012-preprint_30sett2011}, and therefore Corollary \ref{cor:from_Minlos} can be reinterpreted by saying that the space of solutions to
\begin{equation}\label{eq:(T+a)Xi=0_L2-bis}
\qquad(T_\lambda^{(1)}+\alpha\mathbbm{1})\,f\;=\;0\,,\qquad\qquad f\in L^2_{\ell=1}(\mathbb{R}^3)
\end{equation}
(see \eqref{eq:(T+a)Xi=0_L2}) is trivial for $m\geqslant m^{**}_{\textsc{m}}$ and three-dimensional for $m\in(m^*,m^{**}_{\textsc{m}})$.

Observe also that the deficiency index of $T_{\lambda}^{(1)}$ and $T_{\lambda=0}^{(1)}$ clearly coincide (for, as operators in $L^2_{\ell=1}(\mathbb{R}^3)$, they only differ by a bounded and self-adjoint operator) and that the condition $m\in(m^*,m^{**}_{\textsc{m}})$ corresponds to $s(m)\in(0,\frac{1}{2})$, which is precisely the condition that makes the functions $\Xi'_{1,n}$ to belong to $L^2_{\ell=1}(\mathbb{R}^3)$. Of course one cannot expect that the $\Xi'_{1,n}$'s span also $\ker(T_{\lambda}^{(1)*}+\alpha\mathbbm{1})$, however, it is reasonable to expect that the functions in $\ker(T_{\lambda}^{(1)*}+\alpha\mathbbm{1})$ have the same singularity $|p|^{-2+s(m)}$ as $|p|\to +\infty$ (in Fourier transform).

In our discussion in Section \ref{sec:ext_scheme_for_A} we showed that what actually parametrises the self-adjoint realisations of the Ter-Martirosyan--Skornyakov Hamiltonian at given inverse scattering length $\alpha$ are the self-adjoint extensions of the operator $\mathcal{A}_{\lambda,\alpha}^{(1)}$ in the Hilbert space $H^{-1/2}_{W_\lambda,\ell=1}(\mathbb{R}^3)$, which are in turn parametrised by self-adjoint operators on $\ker(\mathcal{A}_{\lambda,\alpha}^{(\ell)})^\star\subset H^{-1/2}_{W_\lambda,\ell}(\mathbb{R}^3)$. We also showed (Proposition \ref{lem:kerAstar}) that $\ker(\mathcal{A}_{\lambda,\alpha}^{(\ell)})^\star$, when $\ell=1$, is the space of distributional solutions to
\begin{equation}\label{eq:(T+a)Xi=0_H-1/2-bis}	
\qquad(T_\lambda^{(1)}+\alpha\mathbbm{1})\,\Xi\;=\;0\,,\qquad\qquad \Xi\in H^{-1/2}_{\ell=1}(\mathbb{R}^3)
\end{equation}
(see \eqref{eq:(T+a)Xi=0_H-1/2}).

Thus, at least for $m<m^{**}_{\textsc{m}}$, \emph{and possibly up to a larger threshold than Minlos's $m^{**}_{\textsc{m}}$}, the equation \eqref{eq:(T+a)Xi=0_H-1/2-bis} must have a space of solutions of dimension three or larger. This is the initial motivation for our conjecture: the possibility of a wider (than Minlos's) range of masses for the multiplicity of self-adjoint TMS operators.

By analogy, it would then be plausible to expect that \eqref{eq:(T+a)Xi=0_H-1/2-bis} has non-trivial solutions, with the same singularity $|p|^{-2+s(m)}$ as $|p|\to +\infty$, on the whole range of mass $m$ for which the corresponding $s(m)$ makes the singularity $|p|^{-2+s(m)}$ a $H^{-\frac{1}{2}}$-singularity. The latter range of masses is exactly $(m^*,m^{**})$.

In fact, one is driven to such a conclusion not only by mere analogy. If one applies $T_\lambda^{(1)}+\alpha\mathbbm{1}$ to a charge $\Xi\in H^{-1/2}_{\ell=1}(\mathbb{R}^3)$ of the form%\footnote{The cutoff for large $p$, similarly to our choice of $\lambda>0$ to have a strictly positive bottom, is only a technical detail which is a consequence of our choice of working within the framework of KVB theory. Indeed, this can be avoided by means of different approaches (as for example in \cite{CDFMT-2015}) for which the ambient space is (Sobolev) homogeneous.}
 \begin{equation}
 \widehat{\Xi}(p)\;=\;\frac{\:\mathbf{1}_{\{|p|\geqslant K\}}\,}{\,|p|^{2-s\,}}\,Y_{1,n}(\Omega_p)
 \end{equation}
for some $K>0$, $s\in(0,1)$, and $n\in\{-1,0,1\}$, and one splits
\[
\begin{split}
(T_\lambda^{(1)}+\alpha\mathbbm{1})\,\Xi\;&=\;T_{\lambda=0}^{(1)}\,\Xi+(T_\lambda^{(1)}-T_{\lambda=0}^{(1)}+\alpha\mathbbm{1})\,\Xi\,,
\end{split}
\]
then it is straightforward to see that 
\[
(T_\lambda^{(1)}-T_{\lambda=0}^{(1)}+\alpha\mathbbm{1})\,\Xi\;\in\;H^{-1/2}_{\ell=1}(\mathbb{R}^3)
\]
and that
\[
\begin{split}
 \widehat{(T_{\lambda=0}^{(1)}\,\Xi)}(p)\;=\;\frac{\,2\pi\,Y_{1,n}(\Omega_p)}{\:|p|^{1-s}}\Big(\pi\sqrt{\nu}\,\mathbf{1}_{\{|p|\geqslant K\}}+\int_{-1}^1\!\ud y\,y\int_{K/\rho}^{+\infty}\!\ud r\,\frac{r^s}{\,r^2+1+\mu ry}\Big)\,.
\end{split}
\]
Now, $T_{\lambda=0}^{(1)}\,\Xi\notin H^{-1/2}_{\ell=1}(\mathbb{R}^3)$ -- whereas surely $T_{\lambda=0}^{(1)}\,\Xi\in H^{-3/2}_{\ell=1}(\mathbb{R}^3)$, consistently with the bound \eqref{eq:T-ell-3/2-1/2_included}. Therefore, for $\Xi$ to be a solution to $(T_\lambda^{(1)}+\alpha\mathbbm{1})\,\Xi=0$, one must have $T_{\lambda=0}^{(1)}\,\Xi=0$, that is, the highest singularity must be cancelled.
Observe that the latter must be a \emph{large momentum} cancellation, independent of the cut-off $|p|\geqslant K$ that we imposed on $\Xi$ at \emph{small} $|p|$: indeed, similarly to the shift $\lambda>0$ needed for the direct application of the Kre{\u\i}n-Vi\v{s}ik-Birman extension formulas, we only cut $|p|\geqslant K$ so as to guarantee $\Xi$ not to be out of $H^{-\frac{1}{2}}(\mathbb{R}^3)$ at $p=0$, but we could have chosen as well to base our discussion on \emph{weak} Sobolev spaces, in analogy of what recently done in \cite{CDFMT-2015} where the quadratic form for the (2+1)-fermionic model was studied by means of homogeneous Sobolev spaces. Thus, we do expect to infer relevant information also when $K\to 0$. Explicitly, $T_{\lambda=0}^{(1)}\,\Xi=0$ in the limit $K\to 0$ reads
\[
\pi\sqrt{\nu}+\int_{-1}^1\!\ud y\,y\int_{0}^{+\infty}\!\ud r\,\frac{r^s}{\,r^2+1+\mu ry}\;=\;0\,,
\]
which is precisely the integral equation \eqref{eq:s-integral-equation} that fixes uniquely $s=s(m)\in(0,1)$ in terms of $m\in(m^*,m^{**})$.

This gives a strong evidence that for $m\in(m^*,m^{**})$ there exist non-zero $H^{-1/2}_{\ell=1}$-solutions to
$(T_\lambda^{(1)}+\alpha\mathbbm{1})\,\Xi=0$, whose singularity in momentum is $\sim|p|^{-2+s(m)}$ as $|p|\to +\infty$, and for $m\geqslant m^{**}$ there are instead none. Combining this with the three-fold dimensionality of the space of solutions on $L^2_{\ell=1}(\mathbb{R}^3)$ when $m\in (m^*,m^{**}_{\textsc{m}})$, we are lead to our main conjecture in Section \ref{sec:ext_scheme_for_A}

A further consistency argument is provided by the fact that in the physical literature it is well known that the Ter-Martirosyan--Skornyakov Hamiltonian, defined via non-rigorous reasoning and physical heuristics, is unambiguous only upon the specification of a further `three-body parameter', needed precisely in the range $(m^*,m^{**})$, and not only in the smaller range $(m^*,m^{**}_{\textsc{m}})$ found by Minlos. In fact, the explicit singularity $|p|^{-2+s(m)}$ of the charges results, via the argument presented in Section \ref{sec:3bodycond}, in a three-body asymptotics at the triple coincidence point that is precisely of the same form found in the physical literature.

%if there is a mass threshold below which 

%{eq:(T+a)Xi=0_L2-bis}

%Qui il discorso che se abbassiamo l'analisi di Minlos (i.e., la sua formula finale) a $H^{-1/2}$ esce proprio la nostra congettura con le masse $m^{**}_{\ell}$ dei fisici.

\appendix

\section{Boundedness properties of $T^{(\ell)}_\lambda$ for $\ell\geqslant 1$}\label{appendix:bounds_ell_-1/2_3/2}

In this Appendix we prove the bound \eqref{eq:T-ell-3/2-1/2_included} of Proposition \ref{prop:T-W}.

Since the multiplicative part of $T_\lambda$ is obviously continuous from $H^s(\mathbb{R}^3)$ to $H^{s-1}(\mathbb{R}^3)$ for any $s\in\mathbb{R}$, it is clear that \eqref{eq:T-ell-3/2-1/2_included} is equivalent to 
\begin{equation}\label{eq:T-ell-3/2-1/2_included_off-diag}
\begin{split}
&\Big\|\int_{\mathbb{R}^3}\ud q\:\frac{\widehat{\xi}(q)}{p^2+q^2+\mu p\cdot q+\lambda}\Big\|_{H^{s-1}}\;\lesssim\;\|\xi\|_{H^{s}} \\
&\qquad\quad \forall \xi\in H_{\ell}^{s}(\mathbb{R}^3)\quad 
\begin{cases}
\;s\in\,\textstyle{[-\frac{1}{2},\frac{3}{2}]} \\
\quad \ell\geqslant 1\,.
\end{cases}
%\qquad \forall \xi\in H_{\ell}^{s}(\mathbb{R}^3)\,.
\end{split}
\end{equation}

In fact, the case $s\in(-\frac{1}{2},\frac{3}{2})$ for any $\ell\geqslant 0$, as well as the case $s=\frac{3}{2}$ for $\ell\geqslant 1$ are already covered in \cite{MO-2016}, as we quoted in the statement of Proposition \ref{prop:T-W}. Thus, we only need to include the case $s=-\frac{1}{2}$, $\ell=1$: we prove it here for completeness, because in \cite{MO-2016} precisely this case was not worked out explicitly.

To this aim, mimicking the proof of \cite[Proposition 4]{MO-2016}, we observe that the integral operator appearing in the l.h.s.~of \eqref{eq:T-ell-3/2-1/2_included_off-diag} acts non-trivially only on the radial part of $\widehat{\xi}$, and precisely as
\begin{equation}\label{eq:kernel qQ_l}
(\mathcal{Q}_{\lambda}^{(\ell)}f)(r)\;=\;2\pi\int_{-1}^{+1}\!\ud y P_\ell(y)\!\int_0^{+\infty}\!\!\frac{f(r')}{r^2+r'^2+\mu r r' y +\lambda}\,r'^2\ud r'\,,
\end{equation}
where
\begin{equation}\label{eq:P-Legendre}
P_\ell(y)\;=\;\frac{1}{2^\ell \ell!}\,\frac{\ud^\ell}{\ud y^\ell}(y^2-1)^\ell
\end{equation}
is the $\ell$-th Legendre polynomial. Thus, proving \eqref{eq:T-ell-3/2-1/2_included_off-diag} is equivalent to proving
\begin{equation}
\|(1+r^2)^{-\frac{3}{4}}(\mathcal{Q}_{\lambda}^{(\ell)}f)\|_{ L^2(\mathbb{R}^+,r^2\,\ud r)}\;\lesssim\;\|(1+r^2)^{-\frac{1}{4}}f\|_{ L^2(\mathbb{R}^+,r^2\,\ud r)}\,,
\end{equation}
which is in turn equivalent, setting $h(r):=rf(r)(1+r^2)^{-\frac{1}{4}}$, to the boundedness in $L^2(\mathbb{R}^+,\ud r)$ of the integral operator $h\mapsto \widetilde{\mathcal{Q}}_{\lambda}^{(\ell)} h$ defined by
\begin{equation}\label{eq:kernel qQtilde_l}
(\widetilde{\mathcal{Q}}_{\lambda}^{(\ell)}h)(r)\;:=\;\int_{-1}^{+1}\!\ud y P_\ell(y)\!\int_0^{+\infty}\!\!\frac{rr'\,(1+r'^2)^{\frac{1}{4}}\,h(r')}{(r^2+r'^2+\mu r r' y +\lambda)(1+r^2)^{\frac{3}{4}}}\,\ud r'\,.
\end{equation}
Using \eqref{eq:P-Legendre} and integrating by parts  $\ell\geqslant 1$ times in $y$ yields
\[
\begin{split}
(\widetilde{\mathcal{Q}}_{\lambda}^{(\ell)}h)(r)\;=\;\frac{(-1)^\ell}{2^\ell\ell!}\int_0^{+\infty}\!\!\!\ud r'\,\frac{rr'\,(1+r'^2)^{\frac{1}{4}}h(r')}{(1+r^2)^{\frac{3}{4}}}\int_{-1}^{+1}\!\ud y\,\frac{(y^2-1)^\ell(\mu r r')^\ell}{(r^2+r'^2+\mu r r' y +\lambda)^{\ell+1}}\,.
\end{split}
\]
Since $|y|\leqslant 1$, analogously to \eqref{eq:lambda-equiv-1}
\begin{equation}\label{eq:lambda-y-equiv-1}
(r^2+r'^2+\mu r r' y +\lambda)\;\sim\;(r_1^2+r_2^2+1)\;\geqslant 0\,.
\end{equation}
Then
\[
\begin{split}
|(\widetilde{\mathcal{Q}}_{\lambda}^{(\ell)}h)(r)|\;&\lesssim\;\int_0^{+\infty}\!\!\!\ud r'\,\frac{rr'\,(1+r'^2)^{\frac{1}{4}}\,|h(r')|}{(1+r^2)^{\frac{3}{4}}}\int_{-1}^{+1}\!\ud y\,\frac{(\mu r r')^\ell}{(r^2+r'^2+\mu r r' y +\lambda)^{\ell+1}} \\
&=\;\int_0^{+\infty}\!\!\!\ud r'\,\frac{rr'\,(1+r'^2)^{\frac{1}{4}}\,|h(r')|}{\ell\,(1+r^2)^{\frac{3}{4}}}\,(\mu r r')^{\ell-1}\,\times \\
&\qquad\qquad\qquad\times\Big(\frac{1}{(r^2+r'^2-\mu r r' y +\lambda)^{\ell}}-\frac{1}{(r^2+r'^2+\mu r r' y +\lambda)^{\ell}}\Big) \\
&\lesssim\;\int_0^{+\infty}\!\!\!\ud r'\,\frac{rr'\,(1+r'^2)^{\frac{1}{4}}\,|h(r')|}{(1+r^2)^{\frac{3}{4}}}\,(\mu r r')^{\ell}\,\times \\
&\qquad\qquad\qquad\times\,\frac{(r^2+r'^2+1)^{\ell-1}}{(r^2+r'^2-\mu r r' y +\lambda)^{\ell}(r^2+r'^2+\mu r r' y +\lambda)^{\ell}} \\
&\lesssim\;\int_0^{+\infty}\!\!\!\ud r'\,\frac{(rr')^{\ell+1}(1+r'^2)^{\frac{1}{4}}}{(1+r^2)^{\frac{3}{4}}(r^2+r'^2+1)^{\ell+1}}\,|h(r')|\;\equiv\;\int_0^{+\infty}\!\!\mathcal{K}_\lambda^{(\ell)}(r,r')\,|h(r')|\,\ud r'\,,
\end{split}
\]
where in the first and last step we used \eqref{eq:lambda-y-equiv-1}, and in the third step we used the formula $(a^\ell-b^\ell)=(a-b)\sum_{j=0}^{n-1}a^{n-j-1}b^{\,j}$ ($a,b\geqslant 0$) .
From
\[
\int_0^{+\infty}\!\!\frac{r'^{\ell+1}(1+r'^2)^{\frac{1}{4}}}{(r^2+r'^2+1)^{\ell+1}}\,\ud r'\;\lesssim\;(1+r'^2)^{-\frac{2\ell-1}{4}}
\]
we deduce
\[\tag{*}
\sup_{r>0}\int_0^{+\infty}\!\!\mathcal{K}_\lambda^{(\ell)}(r,r')\,\ud r'\;=\;\sup_{r>0}\frac{r^{\ell+1}}{(1+r^2)^{\frac{3}{4}}}\int_0^{+\infty}\!\!\frac{r'^{\ell+1}(1+r'^2)^{\frac{1}{4}}}{(r^2+r'^2+1)^{\ell+1}}\,\ud r'\;\lesssim\;1\,,
\]
and from
\[
\int_0^{+\infty}\!\!\frac{r^{\ell+1}}{(1+r^2)^{\frac{3}{4}}(r^2+r'^2+1)^{\ell+1}}\,\ud r\;\lesssim\;(1+r^2)^{-\frac{2\ell+3}{4}}
\]
we deduce
\[\tag{**}
\begin{split}
\sup_{r'>0}\int_0^{+\infty}\!\!&\mathcal{K}_\lambda^{(\ell)}(r,r')\,\ud r\;= \\
&=\;\sup_{r'>0}\:r'^{\ell+1}(1+r'^2)^{\frac{1}{4}}\!\int_0^{+\infty}\!\!\frac{r^{\ell+1}}{(1+r^2)^{\frac{3}{4}}(r^2+r'^2+1)^{\ell+1}}\,\ud r\;\lesssim\;1\,.
\end{split}
\]
A standard Schur test based on (*) and (**) implies $\|\widetilde{\mathcal{Q}}_{\lambda}^{(\ell)}h\|_2\lesssim\|h\|_2$, thus concluding the proof.

\section*{Acknowledgements}

\noindent We warmly thank S.~Albeverio, S.~Becker, G.~Dell'Antonio, and K.~Yajima for many stimulating and inspiring discussions.

\def\cprime{$'$}

\end{document}